\documentclass[a4paper,12pt]{article}
\usepackage{amsmath}
\usepackage{latexsym}
\usepackage{amssymb,amscd}
\usepackage{amsthm}
\usepackage{amsfonts}
\usepackage[dvipdfm]{graphicx}
\usepackage{bm}
\usepackage{cite}

\newcommand{\diag}[1]{\mathrm{diag}\left({#1}\right)}

\newcommand\U[1]{\mathrm{U}{(#1)}}

\newcommand\SU[1]{\mathrm{SU}{(#1)}}

\DeclareMathOperator{\tr}{tr}
\DeclareMathOperator{\Ad}{Ad}

\DeclareMathOperator{\rank}{rank}

\DeclareMathOperator{\Span}{span}

\renewcommand{\proofname}{\bf proof}

\numberwithin{equation}{section}
\numberwithin{theorem}{section}

\title{
  {\bf 
    Bipartite entanglement and control in 
    multiqubit systems\footnote{
      {\bf Key words}:
      bipartite entanglement, control, 
      local transformation group, NMR system.
    }
  }
}
\author{
  Toshihiro IWAI\footnote{iwai[[at mark]]amp.i.kyoto-u.ac.jp}
  \,\, and \,\, 
  Yoshiro YABU\footnote{yoshiro[[at mark]]amp.i.kyoto-u.ac.jp}\\
  {\normalsize Department of Applied Mathematics and Physics}\\
  {\normalsize Kyoto University, Kyoto-606-8501, Japan}
}
\date{}

\newtheorem{Thm}{Theorem}[section]

\newtheorem{Prop}[Thm]{Proposition}
\newtheorem{Lem}[Thm]{Lemma}
\newtheorem{Cor}[Thm]{Corollary}
\newtheorem*{Rem}{Remark}

\newcommand{\Lie}{\mathrm{Lie}}

\begin{document}
\baselineskip=7mm plus 0.2mm minus 0.3mm

\maketitle

\begin{abstract}
In this paper, the effect of the control on bipartite entanglement 
is discussed from a geometric viewpoint 
for a nuclear magnetic resonance (NMR) system as 
a model of the $n$-qubit control system. 
The Hamiltonian of the model is the sum of the drift and control Hamiltonians, 
each of which describes the interaction between pairs of 
qubits (or $\frac{1}{2}$-spins) and 
between one of qubits and an external magnetic field, respectively. 
According to the bipartite partition 
$(\mathbb{C}^{2})^{\otimes n} = 
(\mathbb{C}^{2})^{\otimes\ell}\otimes(\mathbb{C}^{2})^{\otimes m}$ 
with $\ell+m = n$, 
the Schr{\"o}dinger equation for the NMR system is put in the matrix form. 
This paper gives a solution to the Schr{\"o}dinger equation 
with the assumptions of small coupling among qubits and of constant controls. 
The solution is put in the form of power series in small parameters. 
In particular, in the case of the two-qubit NMR system,  
the drift and control Hamiltonians are shown to be coupled to work for 
entanglement promotion, 
by examining solutions to the Schr\"odinger equation in detail. 
The concurrence, a measure of entanglement, is evaluated along the solution 
for a small time interval in order to observe that 
the control effect appears, not at the first-order terms in $t$, 
but at the higher-order terms in $t$. 
The evaluated concurrence also suggests which control makes the two-qubit 
more entangled or less. 
\end{abstract}

\newpage
\section{Introduction}
Entanglement plays many roles in quantum computation and 
quantum information theory, and 
is a resource for quantum communication. 
Recently, the geometry of entanglement has been developed 
\cite{Ema04,Iw07-1,Iw07-2,IHM08,Le04,Lev05,Le05,Le06,MD01}. 
The algebraic study of entanglement also has been made to seek for 
invariants that quantify multiqubit entanglement 
\cite{BLT03,GRB,Sud}. 
The geometry of entanglement in few-qubit systems was studied 
in \cite{MD01, Le04} and in \cite{Le05,Le06} in terms of the Hopf fibration 
and the complex projective spaces, respectively. 
The first author \cite{Iw07-1} showed that 
the concurrence for the two-qubit system 
can be characterized as the coordinate of the factor space of 
the state space by the action of the local transformation group. 
He continued to study bipartite entanglement in 
multiqubit systems \cite{Iw07-2}. 
In particular, the sets of separable and maximally entangled states in 
bipartite entanglement are identified, 
and the concurrence is extended to a measure of 
bipartite entanglement in multiqubit systems. 
Further, the distance between the set of states of prescribed entanglement and 
the set of separable states is measured. 
The bipartite entanglement is measured also for 
the process of Grover's search algorithm \cite{IHM08}. 
For the bipartite entanglement, see also \cite{Ema04,Lev05}. 

On the other hand, from a practical standpoint, 
it is of much importance to control quantum states 
in order to create entangled states. 
Many of papers on control problems for multiqubit systems 
deal with control systems on unitary groups 
\cite{AD'A02,DA00,AD'A03,D'A01,Alt02,KBG01,KG01,KGB02,RSD95}. 
The theory of control on compact Lie groups has been established \cite{JS72}. 
A sufficient condition for controllability 
of the control system on $\SU{N}$ was given in \cite{Alt02}. 
The Cartan decomposition of the Lie algebra and the symmetric space 
are techniques frequently used for studying control systems 
\cite{KBG01,KG01,KGB02}. 

Two notions as to controllability of multiqubit systems were 
defined and studied by Albertini and D'Alessandro \cite{AD'A02, AD'A03}, 
which are ``pure state controllability'' and 
``equivalent state controllability''. 
In \cite{AD'A03}, equivalent state controllability was shown 
to be equivalent to pure state controllability. 
In the case of multiqubit systems lying in electro-magnetic fields, 
they \cite{AD'A02} gave a necessary and sufficient condition for 
the system to be (pure state) controllable. 
The condition is described in terms of the graph associated 
with the interactions of qubits, 
which is related with the Lie algebraic structure of the system. 
Their method is based on the control theory on compact Lie groups 
\cite{JS72, D'A01}. 
As will be pointed out in Appendix B, 
their model Hamiltonian is a bit different from that adopted in this article. 

The aim of this article is to study entanglement and control of 
nuclear magnetic resonance (NMR) multiqubit systems 
from a geometric viewpoint. 
Since the notion of entanglement cannot be defined in terms of 
the whole transformation group, 
the control theory on compact Lie groups will not fit in with the study 
of entanglement and control. 
The state space of a finite-dimensional quantum system is a finite 
complex vector space with a constraint required for the probability density. 
In this article, the state space is a linear space of complex matrices 
with the constraint from the probability density. 
The reason for the choice of the present state space is that 
the entanglement measure like the concurrence is easy to treat 
on this state space. 
The state space admits the two-sided action of the local transformation group 
$\U{2^{\ell}}\otimes \U{2^m}$, where $m$ and $\ell$ are determined by 
the bipartite partition of the multiqubit 
$(\mathbb{C}^{2})^{\otimes n} = 
(\mathbb{C}^{2})^{\otimes\ell}\otimes(\mathbb{C}^{2})^{\otimes m}$. 
As will be stated later, the entanglement of the multiqubit system 
is invariant by the action of the local transformation group. 
Hence, in order to make the multiqubit more entangled, the control 
is required to move the state in the direction transverse to the 
orbit of the local transformation group. 
However, for NMR systems to be treated in this article, 
the control vectors generated by the control Hamiltonian are tangent 
to the orbit mentioned above, so that they would not make the system 
more entangled, if no coupling were made with other vector fields. 
Since the total Hamiltonian of the NMR system is the sum of the control and 
drift Hamiltonians, one may expect the coupling between them to occur. 
The objectives in this article are to set up the control problem 
on the above-mentioned state space and to show, by solving the 
Shr{\"o}dinger equation with constant controls, that the control and 
drift Hamiltonian are coupled together to give rise 
to vector fields transverse to the group orbit for 
the entanglement promotion of the multiqubit system.

The organization of this paper is as follows: 
Section 2 contains a geometric setting up of $n$-qubit systems. 
According to the bipartite partition
$(\mathbb{C}^{2})^{\otimes n} = 
(\mathbb{C}^{2})^{\otimes\ell}\otimes(\mathbb{C}^{2})^{\otimes m}$ 
with $\ell + m = n$, 
the state space $M$ is viewed as the space of $2^{\ell}\times 2^{m}$ 
normalized matrices $C$, 
on which the local transformation group $G=\U{2^{\ell}}\otimes\U{2^{m}}$ acts. 
The action of $G$ leaves invariant the bipartite entanglement 
of the $n$-qubit. 
Linear operators on the $n$-qubit system are mapped to vector 
fields on the state space $M$. 
Lie brackets among those vector fields are given explicitly. 
In Section 3, the Hamiltonian of an NMR model is given, 
which is split into two operators, 
the drift Hamiltonian $\hat{H}_{d}$ and 
control Hamiltonian $\hat{H}_{c}$. 
The $\hat{H}_{d}$ and $\hat{H}_{c}$ describe 
the interactions between pairs of $n$ qubits and 
between one of qubits and an external magnetic field, respectively. 
In addition, a solution to $n$-qubit NMR systems with a certain type of 
drift Hamiltonians is given with the assumption that the 
coupling constants between qubits are small enough. 
The solution is put in the form of power series in small parameters 
related to the coupling constants. 
It will be shown that the vector fields associated with 
the drift and control Hamiltonians are coupled actually. 
In Section 4, the coupling between those vector fields is extensively studied 
on the two-qubit system to observe that entanglement is actually promoted. 
Discussion runs as follows: 
A solution $C(t)$ of the Schr\"odinger equation in the matrix form is 
obtained as a power series in the small coupling constant. 
The solution is examined in detail to show that the vector fields 
associated with the control and drift Hamiltonians are coupled indeed 
to induce vector fields transverse to the $G$-orbit. 
The effect of the control on entanglement is verified by estimating 
the difference between the concurrences of the solutions for 
the controlled and uncontrolled systems for a small time interval 
in the case that the initial state is a diagonal matrix 
$\Lambda = \diag{\lambda_{1}, \lambda_{2}}$. 
The result is that the effect of the control is so slow that 
it emerges at the third-order term in $t$. 
Concluding remarks are made in Section 5, in which 
the Cartan decomposition method is related to the present method. 
Appendix A gives a list of Lie brackets among vector fields 
related with the drift and control Hamiltonians for  
the two-qubit NMR system. 
Appendix B contains the proof of the controllability of 
the $n$-qubit NMR system. 
The method for proof is different from that in \cite{AD'A02}.

\section{Geometric setting for bipartite entanglement}
\label{Geometry of bipartite entanglement}

\subsection{The state space}
We start with geometric setting for the $n$-qubit system 
after \cite{Iw07-2}. 
The Hilbert space for the $n$-qubit system is 
$(\mathbb{C}^{2})^{\otimes n} = 
\underbrace{
  \mathbb{C}^{2}\otimes\cdots\otimes\mathbb{C}^{2}
}_{n\text{ times}}$, 
whose state vectors $|\Psi\rangle$ are put in the form 
\begin{equation*}
  |\Psi\rangle 
  = 
  \sum_{J=(j_{1},\cdots,j_{n})\in\{0,1\}^{n}} 
  c_{J}|j_{1}\rangle \otimes\cdots\otimes |j_{n}\rangle 
  \quad \text{ with } \quad 
  \sum_{J} |c_{J}|^{2} = 1, 
\end{equation*}
where $\{ |0\rangle, |1\rangle \}$ denotes the computational basis of 
a single qubit system $\mathbb{C}^{2}$. 
We denote by $\{ |J\rangle \}_{J\in\{0,1\}^{n}}$ 
the standard orthonormal basis 
$|j_{1}\rangle \otimes\cdots\otimes |j_{n}\rangle$ 
of $(\mathbb{C}^{2})^{\otimes n}$. 
The Hermitian inner product of 
$|\Psi\rangle = \sum_{J}c_{J}|J\rangle$ and 
$|\Phi\rangle = \sum_{J} d_{J}|J\rangle$ is then given by 
$\langle \Psi | \Phi \rangle = \sum_{J} \overline{c_{J}}d_{J}$. 

According to a bipartite partition of the $n$ qubits, 
the Hilbert space $(\mathbb{C}^{2})^{\otimes n}$ is decomposed into the 
tensor product 
$(\mathbb{C}^{2})^{\otimes \ell}\otimes (\mathbb{C}^{2})^{\otimes m}$ 
with $\ell+m=n$. 
We here assume that $0< \ell \leq m$ without loss of generality. 
A state vector $|\Psi\rangle$ of the $n$-qubit system is 
separable with respect to the bipartite decomposition 
$(\mathbb{C}^{2})^{\otimes n} = 
(\mathbb{C}^{2})^{\otimes\ell}\otimes(\mathbb{C}^{2})^{\otimes m}$, 
if there are two unit vectors 
$|\Psi_{1}\rangle \in (\mathbb{C}^{2})^{\otimes\ell}$ and 
$|\Psi_{2}\rangle \in (\mathbb{C}^{2})^{\otimes m}$ such that 
$|\Psi\rangle = |\Psi_{1}\rangle \otimes |\Psi_{2}\rangle$. 
A vector $|\Psi\rangle$ is said to be entangled if it is not separable. 

According to the bipartite decomposition 
$(\mathbb{C}^{2})^{\otimes n} = 
(\mathbb{C}^{2})^{\otimes\ell}\otimes(\mathbb{C}^{2})^{\otimes m}$, 
a state vector $|\Psi\rangle$ can be rewritten in the form 
\begin{equation*}
  |\Psi\rangle 
  = 
  \sum_{A\in\{0,1\}^{\ell}}\sum_{B\in\{0,1\}^{m}} 
  c_{AB}|A\rangle \otimes |B\rangle, 
\end{equation*}
where $A=(j_{1},\cdots,j_{\ell})$ and $B=(j_{\ell+1},\cdots,j_{n})$. 
This expression gives rise to the isomorphism of 
$(\mathbb{C}^{2})^{\otimes n}$ to 
the vector space $\mathbb{C}^{2^{\ell}\times 2^{m}}$ of 
$2^{\ell}\times 2^{m}$ complex matrices, 
\begin{equation}
  \iota \,:\, 
  (\mathbb{C}^{2})^{\otimes n} 
  \longrightarrow 
  \mathbb{C}^{2^{\ell}\times 2^{m}} 
  \, ; \, 
  |\Psi\rangle = \sum_{A, B}c_{AB}|A\rangle\otimes|B\rangle 
  \longmapsto 
  C =  (c_{AB}). 
  \label{iota}
\end{equation}
The $\mathbb{C}^{2^{\ell}\times 2^{m}}$ is endowed with 
the inner product given by $\tr(C_{1}^{\ast} C_{2})$ 
for  $C_{1}, C_{2} \in \mathbb{C}^{2^{\ell}\times 2^{m}}$. 

Since $|\Psi\rangle$ is normalized, 
the corresponding matrix $C=\iota(|\Psi\rangle)$ is subject to 
the condition $\tr(C^{\ast} C) = 1$. 
Thus, the state space for the $n$-qubit system is defined to be 
\begin{equation*}
  M 
  := 
  \{ C \in \mathbb{C}^{2^{\ell}\times 2^{m}} \,|\, \tr(C^{\ast}C) = 1 \}, 
\end{equation*}
which is diffeomorphic with the $(2^{n+1}-1)$-dimensional sphere. 
The state space $M$ becomes a Riemannian manifold equipped with 
the standard metric 
\begin{equation}
  \langle X, Y \rangle_{C} 
  := 
  \frac{1}{2}\tr( X^{\ast}Y + Y^{\ast}X ), 
  \quad 
  X, Y \in T_{C}M, 
  \label{metric}
\end{equation}
where the tangent space to $M$ at $C$ is identified with 
\begin{equation*}
  T_{C}M 
  = 
  \{ 
    X \in \mathbb{C}^{2^{\ell}\times 2^{m}} 
  \,|\, 
    \tr(X^{\ast}C + C^{\ast}X) = 0
  \}.
\end{equation*}

\subsection{Bipartite entanglement}
In this subsection, we make a review of a measure of 
bipartite entanglement introduced by the first author 
\cite{Iw07-1, Iw07-2}. 

The separability condition of a state vector $|\Psi\rangle$ can be 
expressed as $\rank\, CC^{\ast} =1$, 
where $C$ is the matrix corresponding to $|\Psi\rangle$. 
Since the $2^{\ell}\times 2^{\ell}$ matrix $CC^{\ast}$ is 
positive semi-definite together with $\tr(CC^{\ast})=1$ and $\ell\leq m$, 
the condition of $\rank\, CC^{\ast} = 1$ is equivalent to 
$\det ( I_{2^{\ell}} - CC^{\ast} ) = 0$, 
where $I_{2^{\ell}}$ denotes 
the $2^{\ell}\times 2^{\ell}$ identity matrix. 
The first author \cite{Iw07-1, Iw07-2} showed that the quantity 
\begin{equation} 
  F(C) := \det(I_{2^{\ell}}-CC^{\ast}) 
  \label{measurement}
\end{equation} 
serves as a measure of bipartite entanglement, and 
as an extension of the concurrence which is well-known as a measure of 
entanglement in two-qubit system. 
In fact, one has $F(C)=\det(CC^*)$, the square of the concurrence, 
for the two-qubit with $\ell=m=1$. 
The quantity $F(C)$ is non-negative, and vanishes if and only if $C$ is separable. 
Further, the $F$ takes the maximal value 
$(1-1/2^{\ell})^{2^{\ell}}$ 
if and only if all eigenvalues of $CC^{\ast}$ coincide. 
If $F$ attains the maximal value at $C$, 
the matrix $C\in M$ or the corresponding state vector $|\Psi\rangle$ is 
said to be maximally entangled with respect to the bipartite decomposition 
$(\mathbb{C}^{2})^{\otimes n} = 
(\mathbb{C}^{2})^{\otimes\ell}\otimes(\mathbb{C}^{2})^{\otimes m}$. 

\subsection{$\U{2^{\ell}}\otimes\U{2^{m}}$-action}
The action of the group 
$G=\U{2^{\ell}}\otimes\U{2^{m}}$ on 
$M \subset \mathbb{C}^{2^{\ell}\times 2^{m}}$ is defined by 
\begin{equation}
  C \longmapsto g C h^{T}, 
  \label{G-action}
\end{equation}
where $g\in\U{2^{\ell}}$ and $h\in\U{2^{m}}$. 
This action is isometric with respect to \eqref{metric}. 
We remark here that the map $\iota$ is $G$-equivariant, 
namely, $\iota^{-1}(gCh^{T}) = (g\otimes h) \iota^{-1}(C)$. 
As is easily verified, the function $F$ is invariant under 
the $G$-action, so that 
the $G$-action \eqref{G-action} does not change bipartite entanglement. 

We now consider the quotient space $M/G$. 
An arbitrary matrix $C\in M$ is decomposed into the product 
\begin{equation}
  C = g (\Lambda, 0) h^{T}, 
  \quad 
  g \in \U{2^{\ell}}, \, h \in \U{2^{m}}, \, 
  \Lambda=\diag{\lambda_{1},\cdots,\lambda_{2^{\ell}}},  
  \label{singular value decomposition}
\end{equation}
where $0$ is the $2^{\ell}\times(2^{m}-2^{\ell})$ zero matrix, and 
$\lambda_{1}^{2}, \cdots, \lambda_{2^{\ell}}^{2}$ are 
the eigenvalues of $CC^{\ast}$. 
In the following, we assume that 
$\lambda_{1} \geq \cdots \geq \lambda_{2^{\ell}} \geq 0$.  
The quantities $\lambda_{1}, \cdots, \lambda_{2^{\ell}}$ are called 
the singular values of $C$. 
The expression \eqref{singular value decomposition} is equivalent to 
the Schmidt decomposition of a state vector $|\Psi\rangle$. 
From this decomposition, the quotient space $M/G$ can be identified with 
\begin{equation*}
  \left\{
    (\lambda_{1}, \cdots, \lambda_{2^{\ell}}) \in \mathbb{R}^{2^{\ell}} 
  \, \left| \, 
    \lambda_{1} \geq \cdots \geq \lambda_{2^{\ell}} \geq 0, \,\, 
    \lambda_{1}^{2} + \cdots + \lambda_{2^{\ell}}^{2} = 1 
  \right.\right\}. 
\end{equation*}
The sets of separable and maximally entangled states project through 
the natural projection 
$M \to M/G \,;\, C \mapsto (\lambda_{1}, \cdots, \lambda_{2^{\ell}})$ to 
the points $(1,0,\cdots,0)$ and 
$(1/\sqrt{2^{\ell}}, \cdots,
1/\sqrt{2^{\ell}})\in\mathbb{R}^{2^{\ell}}$, 
respectively, both of which lie on the boundary of the quotient $M/G$. 

The state space $M$ is stratified into strata, 
according to $G$-orbit types \cite{Iw07-2}, 
which are determined by the isotropy subgroups 
$G_{C} := \{ g\otimes h \in G \,|\, gCh^{T} = C \}$. 
The isotropy subgroup and the orbits are already studied for 
four-qubit systems in detail, and for multiqubit systems in brief \cite{Iw07-2}.
We now summarize the results in a refined form. 

\begin{Prop}
\label{Prop: isotropy subgrp and orbit}\quad 
Let $N$ be the number of distinct singular values of $C$, and 
$m_{j}$ the multiplicity of its $j$-th largest singular value.

\noindent 
(1) If $CC^{\ast}$ is non-singular, 
the isotropy subgroup $G_{C}$ is isomorphic to 
\begin{equation}
  \big(\U{m_{1}}\times\cdots\times\U{m_{N}}\big) \otimes \U{2^{m}-2^{\ell}}, 
  \label{isotropy subgrp 1}
\end{equation}
where the product $\U{m_{1}}\times\cdots\times\U{m_{N}}$ is 
viewed as a subgroup of $\U{2^{\ell}}$, 
and the $G$-orbit $\mathcal{O}_{C}:=\{(g\otimes h)C | g\otimes h \in G\}$ 
is diffeomorphic to 
\begin{equation}
  G/G_{C} 
   \approx 
  \U{2^{\ell}} 
  \times _{(\U{m_{1}}\times\cdots\times\U{m_{N}})} 
  V_{2^{\ell}}(\mathbb{C}^{2^{m}}), 
  \label{G-orbit 1}
\end{equation}
where $V_{2^{\ell}}(\mathbb{C}^{2^{m}})$ denotes the Stiefel manifold of 
orthonormal $2^{\ell}$ frames in $\mathbb{C}^{2^{m}}$.
\footnote{The symbols $\cong$ and $\approx$ are used to denote the 
isomorphism and the diffeomorphism, respectively.} 

(1-1) In the generic case that singular values of $C$ are all distinct 
and non-zero, $G_{C}$ and $\mathcal{O}_{C}$ take, respectively, the form 
\begin{equation*}
  G_{C} \cong T^{2^{\ell}}\otimes\U{2^{m}-2^{\ell}} 
  \,\,\text{ and }\,\, 
  \mathcal{O}_{C} 
  \approx 
  \U{2^{\ell}}\times_{T^{2^{\ell}}} V_{2^{\ell}}(\mathbb{C}^{2^{m}}). 
\end{equation*}

(1-2) For a maximally entangled state $C$, 
$G_{C}$ and $\mathcal{O}_{C}$ become, respectively, 
\begin{equation*}
  G_{C} \cong \U{2^{\ell}}\otimes\U{2^{m}-2^{\ell}} 
  \,\,\text{ and }\,\,
  \mathcal{O}_{C} \approx V_{2^{\ell}}(\mathbb{C}^{2^{m}}). 
\end{equation*}

\noindent 
(2) If $CC^{\ast}$ is singular, the group $G_{C}$ is isomorphic to 
\begin{equation}
  \big(\U{m_{1}} \times \cdots \times \U{m_{N-1}} \times \U{m_{N}} \big)
  \otimes
  \U{2^{m}-2^{\ell}+m_{N}},  
  \label{isotropy subgrp 2}
\end{equation}
and the orbit $\mathcal{O}_{C}$ is diffeomorphic to 
\begin{equation}
  G/G_{C} 
   \approx 
  V_{2^{\ell}-m_{N}}(\mathbb{C}^{2^{\ell}}) 
  \times _{(\U{m_{1}}\times\cdots\times\U{m_{N-1}})} 
  V_{2^{\ell}-m_{N}}(\mathbb{C}^{2^{m}}).
  \label{G-orbit 2}
\end{equation}

(2-1) In particular, if $C$ is separable, 
$G_{C}$ and $\mathcal{O}_{C}$ are, respectively, of the form 
\begin{equation*}
  G_{C} 
  \cong \U{2^{\ell}-1}\times\U{2^{m}-1}, 
  \,\,\text{ and }\,\, 
  \mathcal{O}_{C} \approx S^{2^{\ell+1}-1}\times_{\U{1}} S^{2^{m+1}-1}.
\end{equation*}
\end{Prop}

\subsection{Vertical and horizontal subspaces}
This subsection deals with the infinitesimal action of 
$G=\U{2^{\ell}}\otimes\U{2^{m}}$. 
For $\xi\in\frak{u}(2^{\ell})$ and $\eta\in\frak{u}(2^{m})$, 
the fundamental vector field associated with 
$\xi\otimes I_{2^{m}} + I_{2^{\ell}}\otimes \eta \in \mathfrak{g}=\Lie(G)$ is 
determined through 
\begin{equation}
  X_{\xi\otimes I_{2^{m}} + I_{2^{\ell}}\otimes\eta}(C)
  := 
  \left.\frac{d}{dt}\right|_{t=0}e^{t\xi}C e^{t\eta^{T}} 
  = 
  \xi C + C\eta^{T}. 
  \label{fundamental vector}
\end{equation}
The vertical subspace $V_C$ is defined to be the tangent space 
$T_C\mathcal{O}_C$ at $C$ to the orbit $\mathcal{O}_C$, and formed by 
the fundamental vector fields evaluated at $C$, 
\begin{equation*}
  V_{C} = 
  \{ 
    X_{\xi\otimes I_{2^{m}} + I_{2^{\ell}}\otimes\eta}(C)
    = 
    \xi C + C\eta^{T} 
  \,|\, 
    \xi\in\frak{u}(2^{\ell}), \eta\in\frak{u}(2^{m}) 
  \}.
\end{equation*} 
The dimension of $V_{C}$ is given by 
$ \dim V_{C} =   2^{2\ell} + 2^{2m} - 1  - \dim G_{C}$. 
In particular, from Proposition \ref{Prop: isotropy subgrp and orbit}, 
we have 
\begin{equation}
  \dim V_{C} = 
  \begin{cases} 
    2^{\ell+1} + 2^{m+1} - 3, & \text{ if $C$ is separable, } \\
    2^{\ell+m+1} - 2^{2\ell},      & \text{ if $C$ is maximally entangled. } 
  \end{cases}
  \label{dim V}
\end{equation}

The horizontal subspace $H_{C}$ of $T_{C}M$ is defined to be 
the orthogonal complement of $V_{C}$ with respect to 
the metric \eqref{metric}. 
A vector $X\in T_{C}M$ is horizontal, if and only if 
$\langle X,\xi C+C\eta^T\rangle=0$ for any $\xi,\eta \in\frak{u}(2^{\ell})$, 
so that the horizontal subspace at $C\in M$ is expressed as 
\begin{equation*}
  H_{C} = 
  \{ 
    X \in T_{C}M 
  \,|\, 
    XC^{\ast} - CX^{\ast} = 0, C^{\ast}X - X^{\ast}C = 0 
  \}. 
\end{equation*}
From \eqref{dim V} together with $\dim H_{C} = 2^{\ell+m+1} - 1 - \dim V_{C}$, 
we have 
\begin{align*}
  \dim H_{C} 
  &= 
  \begin{cases} 
    2(2^{\ell} - 1)(2^{m} - 1), & \text{ if $C$ is separable, } \\
    2^{2\ell} -1,      & \text{ if $C$ is maximally entangled. } 
  \end{cases}
\end{align*}

\subsection{The Lie algebra of linear vector fields on $M$}
In discussing control problems, we need to calculate Lie brackets 
of vector fields on $\mathbb{C}^{2^{\ell}\times 2^{m}}$. 
In particular, we have to work with the Lie brackets of vector fields of 
the form 
\begin{equation*}
  X_{A\otimes B}(C) := ACB^{T}, 
  \quad 
  C \in \mathbb{C}^{2^{\ell}\times 2^{m}},   
\end{equation*}
where $A \in \mathbb{C}^{2^{\ell}\times 2^{\ell}}$ and 
$B \in \mathbb{C}^{2^{m}\times 2^{m}}$ do not need to be Hermitian or 
skew Hermitian at present. 
The Lie bracket of 
$A_{1}\otimes B_{1}, A_{2}\otimes B_{2} \in 
\mathbb{C}^{2^{\ell}\times 2^{\ell}}\otimes\mathbb{C}^{2^{m}\times2^{m}}$ 
is defined to be  
\begin{equation}
  [A_{1}\otimes B_{1}, A_{2}\otimes B_{2}] 
  := 
  [A_{1}, A_{2}]\otimes B_{1}B_{2} + A_{2}A_{1}\otimes[B_{1}, B_{2}]. 
  \label{Lie bracket}
\end{equation}
Since this bracket is verified to satisfy the Jacobi identity, 
the linear space 
$\mathbb{C}^{2^{\ell}\times 2^{\ell}}\otimes \mathbb{C}^{2^{m}\times 2^{m}}$ 
is endowed with a Lie algebraic structure. 
This Lie algebra 
is isomorphic to $\frak{gl}(2^{n},\mathbb{C})$ through 
\begin{equation*} 
  f \, : \, 
  \mathbb{C}^{2^{\ell}\times 2^{\ell}}\otimes\mathbb{C}^{2^{m}\times 2^{m}} 
  \longrightarrow 
  \mathbb{C}^{2^{n}\times 2^{n}} 
  \, ; \, 
  A\otimes B \longmapsto 
  \begin{pmatrix} 
    a_{11}B & \cdots & a_{1\, 2^{\ell}}B \\
    \vdots & & \vdots \\
    a_{2^{\ell}\, 1}B & \cdots & a_{2^{\ell}\, 2^{\ell}}B 
  \end{pmatrix}. 
\end{equation*}
In fact, $f$ is shown to be bijective and homomorphic, 
$f([A_{1}\otimes B_{1}, A_{2}\otimes B_{2}]) = 
[f(A_{1}\otimes B_{1}), f(A_{2}\otimes B_{2})]$. 

Now we are in a position to state the following proposition: 
\begin{Prop}
For $A_{1}\otimes B_{1}, A_{2}\otimes B_{2} \in 
\mathbb{C}^{2^{\ell}\times 2^{\ell}} 
\otimes \mathbb{C}^{2^{m}\times 2^{m}}$, 
the Lie bracket of the associated vector fields 
$X_{A_{1}\otimes B_{1}}$ and $X_{A_{2}\otimes B_{2}}$ on 
$\mathbb{C}^{2^{\ell}\times 2^{m}}$ is given by 
\begin{equation*}
  [X_{A_{1}\otimes B_{1}}, X_{A_{2}\otimes B_{2}}]
  = 
  -X_{[A_{1}\otimes B_{1}, A_{2}\otimes B_{2}]}. 
\end{equation*}
This means that the correspondence 
$A\otimes B \mapsto X_{A\otimes B}$ gives rise to 
an anti-isomorphism of 
$(\mathbb{C}^{2^{\ell}}\otimes \mathbb{C}^{2^{m}}, 
[\bullet, \bullet])$ 
to the Lie algebra 
$\{ X_{A\otimes B} \,|\, 
A \in \mathbb{C}^{2^{\ell}\times 2^{\ell}}, 
B \in \mathbb{C}^{2^{m}\times 2^{m}} \}$ of 
linear vector fields on $\mathbb{C}^{2^{\ell}\times 2^{m}}$. 
\label{Prop: comp of brackets}
\end{Prop}

\begin{Cor}
In the two-qubit system with $\ell=m=1$, one has  
\begin{equation}
\label{2qubit commutators}
  [X_{i\sigma_{j}\otimes\sigma_{k}}, X_{i\sigma_{j'}\otimes\sigma_{k'}}]
  = 
  -X_{[i\sigma_{j}\otimes\sigma_{k}, i\sigma_{j'}\otimes\sigma_{k'}]} 
  = 
  X_{[\sigma_{j},\sigma_{j'}]\otimes\sigma_{k}\sigma_{k'}}
  + 
  X_{\sigma_{j'}\sigma_{j}\otimes[\sigma_{k},\sigma_{k'}]}
\end{equation}
for $j, j', k, k' = 0, 1, 2, 3.$
Here the  $\sigma_{j}, j=1,2,3,$ denote the Pauli matrices defined as 
\begin{equation*}
  \sigma_{1} 
  := 
  \frac{1}{2}\begin{pmatrix} 0 & 1 \\ 1 & 0 \end{pmatrix}, \quad 
  \sigma_{2} 
  := 
  \frac{1}{2}\begin{pmatrix} 0 & i \\ -i & 0 \end{pmatrix}, \quad 
  \sigma_{3} 
  := 
  \frac{1}{2}\begin{pmatrix} 1 & 0 \\ 0 & -1 \end{pmatrix}. 
\end{equation*}
\label{Prop: comp of brackets 2}
\end{Cor}

\renewcommand{\proofname}{
  {\it Proof of Proposition \ref{Prop: comp of brackets}}
}

\begin{proof}
Let $\varphi_{A\otimes B}^{t}$ denote the flow generated by 
the vector field 
$X_{A\otimes B}$ on $\mathbb{C}^{2^{\ell}\times 2^{m}}$. 
We identify $(\mathbb{C}^{2})^{\otimes n}$ with the space 
$\mathbb{C}^{2^{n}}$ of $2^{n}$-dimensional column vectors, 
and the isomorphism 
$\iota : (\mathbb{C}^{2})^{\otimes n} \to 
\mathbb{C}^{2^{\ell}\times 2^{m}}$ 
with an isomorphism of $\mathbb{C}^{2^{n}}$ to 
$\mathbb{C}^{2^{\ell}\times 2^{m}}$. 
We note here that the following diagram commutes, 
\begin{equation*}
  \begin{CD} 
    \mathbb{C}^{2^{\ell}\times 2^{m}} 
    @> \text{action of } A\otimes B >> 
    \mathbb{C}^{2^{\ell}\times 2^{m}} \\
    @V \iota^{-1} VV   @VV \iota^{-1} V \\
    \mathbb{C}^{2^{n}} 
    @>> \text{action of } f(A\otimes B) > 
    \mathbb{C}^{2^{n}}. 
  \end{CD}
\end{equation*}
This means that 
\begin{equation}
\label{X and Y}
  (\iota^{-1})_{\ast}X_{A\otimes B}(C) = Y_{f(A\otimes B)}(\iota^{-1}(C)) 
  \quad \text{and} \quad 
  \iota^{-1}\circ\varphi_{A\otimes B}^{t} = e^{tf(A\otimes B)}\circ\iota^{-1}, 
\end{equation}
where $Y_{f(A\otimes B)}(\bm{c}) = f(A\otimes B)\bm{c}$ with 
$\bm{c} \in \mathbb{C}^{2^{n}}$. 

By the definition of the Lie bracket together with \eqref{X and Y}, 
we obtain 
\begin{align*}
  & 
  [X_{A_{1}\otimes B_{1}}, X_{A_{2}\otimes B_{2}}](C) \\
  = & 
  \left. \frac{\partial^{2}}{\partial t \partial s} \right|_{t=s=0}
  (\varphi_{A_{2}\otimes B_{2}}^{-s}\circ\varphi_{A_{1}\otimes B_{1}}^{-t}\circ 
  \varphi_{A_{2}\otimes B_{2}}^{s}\circ\varphi_{A_{1}\otimes B_{1}}^{t})(C)  \\ 
  = &
  \left. \frac{\partial^{2}}{\partial t \partial s} \right|_{t=s=0}
  \iota\circ 
  \left(
    e^{-sf(A_{2}\otimes B_{2})}\circ e^{-tf(A_{1}\otimes B_{1})}\circ 
    e^{sf(A_{2}\otimes B_{2})}\circ e^{tf(A_{1}\otimes B_{1})} 
  \right)
  \circ\iota^{-1}(C) \\ 
  = &
  \, \iota_{\ast}\, 
  [Y_{f(A_{1}\otimes B_{1})}, Y_{f(A_{2}\otimes B_{2})}](\iota^{-1}(C)). 
\end{align*}  
Further, by the formula 
$[Y_{\Xi_{1}}, Y_{\Xi_{2}}] = - Y_{[\Xi_{1}, \Xi_{2}]}$ for 
any $\Xi_{1}, \Xi_{2} \in \mathbb{C}^{2^{n}\times 2^{n}}$,  
one has 
\begin{align*}
  & \, 
  [X_{A_{1}\otimes B_{1}}, X_{A_{2}\otimes B_{2}}](C) \\
  =& \, 
  \iota_{\ast}\, 
  [Y_{f(A_{1}\otimes B_{1})}, Y_{f(A_{2}\otimes B_{2})}] (\iota^{-1}(C)) 
  = 
  \iota_{\ast}\, 
  (-Y_{[f(A_{1}\otimes B_{1}), f(A_{2}\otimes B_{2})]}) (\iota^{-1}(C)) \\
  =& \,  
  -\iota_{\ast}\, 
  Y_{f([A_{1}\otimes B_{1}, A_{2}\otimes B_{2}])} (\iota^{-1}(C)) 
  =  
  - X_{[A_{1}\otimes B_{1}, A_{2}\otimes B_{2}]}(C). 
\end{align*}  
Since the map $A\otimes B\mapsto Y_{f(A\otimes B)}$ is bijective to 
the space of linear vector fields on $\mathbb{C}^{2^n}$, 
so is the map $A\otimes B\mapsto X_{A\otimes B}$ to the space of 
linear vector fields on $\mathbb{C}^{2^{\ell}\times 2^m}$. 
This ends the proof. 
\end{proof}
\renewcommand{\proofname}{{\it Proof}}


We now apply \eqref{2qubit commutators} for the two-qubit system. 
We take, for instance, 
$A_{1}\otimes B_{1} = i\sigma_{3}\otimes\sigma_{3}$ and 
$A_{2}\otimes B_{2} = i\sigma_{1}\otimes I$. 
By Corollary \ref{Prop: comp of brackets 2}, 
the Lie bracket of the associated linear vector fields is given by 
$  [X_{i\sigma_{3}\otimes\sigma_{3}}, X_{i\sigma_{1}\otimes I}] 
  = 
  X_{[\sigma_{3},\sigma_{1}]\otimes\sigma_{3}}
  + 
  X_{\sigma_{1}\sigma_{3}\otimes[\sigma_{3}, I]} 
  = 
  -X_{i\sigma_{2}\otimes\sigma_{3}}$. 
Other Lie brackets will be found in Appendix A.

\section{Controls in NMR systems}
\label{Control problem for entanglement}
Sections \ref{Control problem for entanglement} and 
\ref{Control problem in two-qubit systems} deal with 
control problems in $n$- and $2$-qubit systems, respectively. 
In this paper, we adopt a nuclear magnetic resonance (NMR) system as
a quantum computation model from $n$-qubit systems 
\cite{KG01}.

The total Hamiltonian operator $\hat{H}$ of the NMR system is 
the sum of two operators, 
\begin{equation*}
  \hat{H} := \hat{H}_{d} + \hat{H}_{c}, 
\end{equation*}
where $\hat{H}_{d}$ and $\hat{H}_{c}$ are called 
the drift and control Hamiltonians, 
and expressed as 
\begin{subequations}
  \begin{gather}
    \hat{H}_{d} 
    := 
    \sum_{1\leq \alpha < \beta \leq n} 
    J_{\alpha\beta} \sigma_{3}^{(\alpha,n)}\sigma_{3}^{(\beta,n)}, 
    \\
    \hat{H}_{c}
    := 
    \sum_{\alpha=1}^{n} \biggl(
      v^{(\alpha)}_{1}\sigma_{1}^{(\alpha,n)} 
      +
      v^{(\alpha)}_{2}\sigma_{2}^{(\alpha,n)} 
    \biggr), 
  \end{gather}
  \label{Hamiltonians}
\end{subequations}
respectively. 
Here, $J_{\alpha\beta}\in\mathbb{R}$ are coupling constants determining 
the strength of the interactions between the $\alpha$-th and $\beta$-th qubits. 
The $v^{(\alpha)}_{1}, v^{(\alpha)}_{2}:\mathbb{R}\to\mathbb{R}$ are 
(time-dependent) controls acting on the $\alpha$-th qubit.
The $\sigma_{j}^{(\alpha,n)}$ are defined as 
\begin{equation}
  \label{sigma}
  \sigma_{j}^{(\alpha,n)}
  := 
  \underbrace{I\otimes \cdots \otimes I}_{(\alpha-1)\text{ times}} 
  \otimes\sigma_{j}\otimes
  \underbrace{I\otimes\cdots\otimes I}_{(n-\alpha)\text{ times}}, 
\end{equation}
for $j=1,2,3,$ and $\alpha = 1, \cdots, n,$ 
with $I$ the $2\times 2$ identity matrix. 
The controls $v^{(\alpha)}_{1}$ and $v^{(\alpha)}_{2}$ are 
assumed to be piecewise constant functions in $t$. 

The Schr\"odinger equation for the NMR system is expressed as 
\begin{equation}
  \frac{d}{dt}|\Psi\rangle 
  = 
  -i\hat{H}|\Psi\rangle, 
  \label{Schrodinger equation 0}
\end{equation}
where the natural unit system has been adopted, so that 
the Plank constant $\hbar$ is equal to one. 
Since we are interested in bipartite entanglement, 
we have to put the Schr\"odinger equation in the form of 
matrix equation. 
According to the isomorphism 
$(\mathbb{C}^{2})^{\otimes n} \cong \mathbb{C}^{2^{\ell}\times 2^{m}}$ 
with $n=\ell+m$, 
the vector fields associated with the drift and control Hamiltonians 
are expressed, respectively, as 
\begin{subequations}
\label{assoc vec fields}
\begin{align}
  -iH_d(C)
  =  &
  -i \!\!\sum_{1\leq \alpha < \beta \leq \ell} \!\!
  J_{\alpha\beta} 
  \, [ \sigma_{3}^{(\alpha,\ell)}\sigma_{3}^{(\beta,\ell)} ] \, C 
  \hspace{3mm} -i \hspace{-5mm}
  \sum_{\ell+1\leq \alpha < \beta \leq \ell+m} \hspace{-5mm}
  J_{\alpha\beta} 
  C \,
  [ \sigma_{3}^{(\alpha,m)}\sigma_{3}^{(\beta,m)} ]^{T} 
  \notag\\ 
  & 
  -i \hspace {-3mm} \sum_{\substack{1\leq \alpha \leq \ell,\\ 
  \ell+1\leq \beta \leq \ell+m}} \hspace{-7mm}
  J_{\alpha\beta} 
  \, [ \sigma_{3}^{(\alpha,\ell)} ]\, C \,[ \sigma_{3}^{(\beta-\ell,m)} ]^{T},     \\ 
  -iH_c(C)
  = &
  \sum_{1 \leq \alpha \leq \ell}
  [ I^{\otimes(\alpha-1)}\otimes\xi_{\alpha}\otimes 
  I^{\otimes(\ell-\alpha)} ]\, C 
  +\sum_{\ell+1 \leq \beta \leq \ell+m} \hspace{-3mm}
  C \,[ I^{\otimes(\beta-\ell-1)}\otimes\xi_{\beta}\otimes 
  I^{\otimes(m-\beta+\ell)} ]^{T},
\end{align}
\end{subequations}
where 
\begin{equation*}
   \xi_{\alpha} = -i(v^{(\alpha)}_{1}\sigma_{1}+v^{(\alpha)}_{2}\sigma_{2}), 
\end{equation*}
and where the symbols with square brackets like 
$[\sigma_{3}^{(\alpha,\ell)}]$ denote the Kronecker products of matrices 
corresponding to the tensor products concerned. 
The Schr{\"o}dinger equation on $M$ is then put in the form 
\begin{equation}
\frac{dC}{dt}=-iH_d(C)-iH_c(C).
\label{Schrodinger eq} 
\end{equation}
As is easily verified, $-iH_d$ and $-iH_c$ are vector fields on $M$. 
In what follows, we will refer to $-iH_d(C)$ and $-iH_c(C)$ as 
the drift and control vector fields, respectively. 

\subsection{A solution to the NMR system}
We show that Eq.~\eqref{Schrodinger eq} can be solved with 
the assumption that the coupling constants are small enough and 
that the controls are constant. 
In view of \eqref{assoc vec fields}, 
we may put the drift and control vector fields in the form, 
\begin{subequations}
\label{N param vect fields}
\begin{align}
  -iH_{d}(C) 
  &= 
  -i\sum_{k=1}^{N}\varepsilon_{k}\Sigma^{(k)}_{1}C(\Sigma^{(k)}_{2})^{T}, \\
  -iH_{c}(C) 
  &= 
  \xi_{1}C + C\xi_{2}^{T}, 
\end{align}
\end{subequations}
respectively, where $\varepsilon_{k}, k=1,\cdots, N,$ are small parameters, 
where $\Sigma^{(k)}_{1}\in\mathbb{C}^{2^{\ell}\times 2^{m}}, 
\Sigma^{(k)}_{2}\in\mathbb{C}^{2^{\ell}\times 2^{m}}, k=1,\cdots, N,$ 
are constant Hermitian matrices, 
and where $\xi_{1}\in\frak{su}(2^{\ell}), \xi_{2}\in\frak{su}(2^{m})$ 
are constant controls. 
The Schr\"odinger equation we are to solve is 
\begin{align}
  \dot{C} 
  &= 
  -i\sum_{k=1}^{N}\varepsilon_{k}\Sigma^{(k)}_{1}C(\Sigma^{(k)}_{2})^{T} 
  + \xi_{1}C + C\xi_{2}^{T} \notag\\
  &= 
  -\sum_{k=1}^{N}\varepsilon_{k}X_{i\Sigma^{(k)}_{1}
    \otimes\Sigma^{(k)}_{2}}(C) 
  + X_{\xi_{1}\otimes I_{2^{m}}}(C) 
  + X_{I_{2^{\ell}}\otimes\xi_{2}}(C) . 
  \label{Schrodinger eq with 2-para}
\end{align}
Since the coupling constants $\varepsilon_{k}$ are sufficiently small, 
we may suppose that the solution $C(t)$ to \eqref{Schrodinger eq with 2-para} 
can be expanded into a power series in $\varepsilon_{k}, k=1,\cdots, N,$ 
\begin{equation*}
  C(t) = 
  \sum_{\bm{n}=(n_{1},\cdots, n_{N})\in(\mathbb{Z}_{\geq 0})^{N}} 
  \varepsilon^{\bm{n}} C_{\bm{n}}(t) 
  \quad \text{with} \quad 
  \varepsilon^{\bm{n}} := 
  \varepsilon_{1}^{n_{1}}\cdots\varepsilon_{N}^{n_{N}}. 
\end{equation*}
By substituting this expansion into \eqref{Schrodinger eq with 2-para} 
and comparing the left and right hand-sides, 
we obtain the series of differential equations 
\begin{equation}
 \label{two-param eq: expansion}
  \begin{aligned}
    \dot{C}_{\bm{0}} 
    &= 
    \hspace{5.2cm} \,\,\,\,\, 
    \xi_{1}C_{\bm{0}} + C_{\bm{0}}\xi_{2}^{T}, \\
    \dot{C}_{\bm{n}} 
    &= 
    -i\sum_{\substack{k=1,\cdots,N, \\ n_{k}\ne 0}}
    \Sigma^{(k)}_{1} C_{\bm{n}-\bm{e}_{k}} (\Sigma^{(k)}_{2})^{T}
    +\xi_{1}C_{\bm{n}} + C_{\bm{n}}\xi_{2}^{T}  
    \quad \text{for} \quad 
    \bm{n} \ne \bm{0}, 
  \end{aligned}
\end{equation}
where 
$\bm{e}_{1}=(1,0,\cdots, 0), \cdots, 
\bm{e}_{N}=(0,\cdots,0,1)\in(\mathbb{Z}_{\geq 0})^{N}$. 

To solve the above equations, we introduce new unknown functions by 
$Q_{\bm{n}}(t) = (e^{-t\xi_{1}}\otimes e^{-t\xi_{2}})\cdot C_{\bm{n}}(t)$. 
Then $Q_{\bm{n}}(t)$ with $\bm{n} \in (\mathbb{Z}_{\geq 0})^{N}$ 
prove to satisfy the equations 
\begin{equation*}
  \begin{aligned}
    \dot{Q}_{\bm{0}} 
    &= 
    0, \\
    \dot{Q}_{\bm{n}} 
    &= 
    -i\sum_{\substack{k=1,\cdots, N, \\ n_{k}\ne 0}}
    \left( e^{-t\xi_{1}} \Sigma^{(k)}_{1} e^{t\xi_{1}} \right) 
    Q_{\bm{n}-\bm{e}_{k}} 
    \left( e^{-t\xi_{2}} \Sigma^{(k)}_{2} e^{t\xi_{2}} \right)^{T} \\
    &= 
    -i\sum_{\substack{k=1,\cdots, N, \\ n_{k}\ne 0}}
    \Ad_{e^{-t\xi_{1}}\otimes e^{-t\xi_{2}}}(\Sigma^{(k)}_{1}\otimes
    \Sigma^{(k)}_{2}) 
    Q_{\bm{n}-\bm{e}_{k}} 
    \quad \text{for} \quad 
    \bm{n} \ne \bm{0}. 
  \end{aligned}
\end{equation*}
These equations for $Q_{\bm{n}}$ are inductively integrable. 
The first (N+1) of solutions are given by  
$Q_{\bm{0}}(t) = Q_{\bm{0}}(0)$ 
and 
\begin{equation*}
  Q_{\bm{e}_{k}}(t) 
  = 
  Q_{\bm{e}_{k}}(0) 
  -i\int_{0}^{t} \,  
  \Ad_{e^{-s\xi_{1}}\otimes e^{-s\xi_{2}}}(\Sigma^{(k)}_{1}\otimes
  \Sigma^{(k)}_{2})  Q_{\bm{0}}(0) ds, \quad k=1,\cdots, N. 
\end{equation*}
In order to write down solutions in a compact form, 
we introduce $N$ linear operators 
$\mathcal{T}_{k} : C^{\infty}(\mathbb{R};\mathbb{C}^{2^{\ell}\times 2^{m}}) 
\to C^{\infty}(\mathbb{R};\mathbb{C}^{2^{\ell}\times 2^{m}}), k=1,\cdots,N,$ 
defined by
\begin{equation}
  (\mathcal{T}_{k}F)(t) 
  := 
  -i\int_{0}^{t}\,  
  \Ad_{e^{-s\xi_{1}}\otimes e^{-s\xi_{2}}}(\Sigma^{(k)}_{1}\otimes
  \Sigma^{(k)}_{2}) F(s) ds, 
  \quad 
  F \in C^{\infty}(\mathbb{R};\mathbb{C}^{2^{\ell}\times 2^{m}}). 
  \label{eq: T_k}
\end{equation}
Then, in terms of these operators, 
$Q_{\bm{e}_{k}}(t), k=1,\cdots, N$, take the form 
\begin{equation*}
  Q_{\bm{e}_{k}}(t) 
  = 
  Q_{\bm{e}_{k}}(0) 
  + 
  (\mathcal{T}_{k}Q_{\bm{0}}(0))(t), 
  \quad 
  k = 1, \cdots, N, 
\end{equation*}
where $Q_{\bm{0}}(0)$ in the right-hand side is viewed as a constant function. 

We proceed to, say, $Q_{\bm{e}_{1}+\bm{e}_{2}}(t)$. 
The equation for $Q_{\bm{e}_{1}+\bm{e}_{2}}(t)$ is expressed as 
\begin{equation*}
  \dot{Q}_{\bm{e}_{1}+\bm{e}_{2}} 
  = 
  -i\Ad_{e^{-t\xi_{1}}\otimes e^{-t\xi_{2}}}(\Sigma^{(2)}_{1}\otimes
  \Sigma^{(2)}_{2})   Q_{\bm{e}_{1}} 
  -i\Ad_{e^{-t\xi_{1}}\otimes e^{-t\xi_{2}}}(\Sigma^{(1)}_{1}\otimes
  \Sigma^{(1)}_{2})   Q_{\bm{e}_{2}},    
\end{equation*}
and integrated to yield 
\begin{align*}
  & \, 
  Q_{\bm{e}_{1}+\bm{e}_{2}}(t) 
  = 
  Q_{\bm{e}_{1}+\bm{e}_{2}}(0)
  + (\mathcal{T}_{2}Q_{\bm{e}_{1}})(t) 
  + (\mathcal{T}_{1}Q_{\bm{e}_{2}})(t) \\
  =& \, 
  Q_{\bm{e}_{1}+\bm{e}_{2}}(0) 
  + (\mathcal{T}_{2}Q_{\bm{e}_{1}}(0))(t) 
  + (\mathcal{T}_{1}Q_{\bm{e}_{2}}(0))(t) 
  + (\mathcal{T}_{2}\circ\mathcal{T}_{1}\, Q_{\bm{0}}(0))(t) 
  + (\mathcal{T}_{1}\circ\mathcal{T}_{2}\, Q_{\bm{0}}(0))(t). 
\end{align*}
In a similar manner, the $Q_{\bm{n}}(t)$ is inductively integrated 
and expressed as 
\begin{equation*}
  Q_{\bm{n}}(t) 
  = 
  \sum_{\substack{
    \bm{m}=(m_{1},\cdots,m_{N})\in(\mathbb{Z}_{\geq 0})^{N}, \\ 
    m_{k} \leq n_{k}, \, k=1, \cdots, N 
  }}\,\, 
  \sum_{\substack{ 
    K = (k_{1}, \cdots, k_{|\bm{m}|}) \in \{1, \cdots, N\}^{|\bm{m}|}, \\
    \#\{\nu | k_{\nu} = k\} = m_{k}, \, k=1,\cdots, N
  }}
  (\mathcal{T}_{K}Q_{\bm{n}-\bm{m}}(0))(t), 
\end{equation*}
where $|\bm{m}| := m_{1} +\cdots +m_{N}$, and where 
$\mathcal{T}_{K}: C^{\infty}(\mathbb{R};\mathbb{C}^{2^{\ell}\times 2^{m}}) 
\to C^{\infty}(\mathbb{R};\mathbb{C}^{2^{\ell}\times 2^{m}})$ 
are defined to be 
\begin{equation*}
  \mathcal{T}_{K} 
  := 
  \mathcal{T}_{k_{1}} \circ \cdots \circ \mathcal{T}_{k_{|\bm{m}|}}, 
  \quad 
  K = (k_{1}, \cdots, k_{|\bm{m}|}) \in \{1, \cdots, N\}^{|\bm{m}|}. 
\end{equation*}
Hence, we have obtained the following. 

\begin{Prop}
\label{Prop: power series in general}
The solution $C(t)$ to \eqref{Schrodinger eq with 2-para} 
is put in the form of power series in 
$\varepsilon_{k}$, $k=1,\cdots, N$, 
\begin{align}
  C(t) 
  &= 
  (e^{t\xi_{1}}\otimes e^{t\xi_{2}}) 
  \sum_{\bm{n}} \sum_{\bm{m}} \sum_{K} 
  \varepsilon^{\bm{n}} (\mathcal{T}_{K} Q_{\bm{n}-\bm{m}}(0))(t)  
  \nonumber \\
  &= 
  (e^{t\xi_{1}}\otimes e^{t\xi_{2}}) 
  \sum_{
    \bm{m}\in(\mathbb{Z}_{\geq 0})^{N}
  }\,\, 
  \sum_{\substack{ 
    K \in \{1, \cdots, N\}^{|\bm{m}|}, \\
    \#\{\nu | k_{\nu} = k\} = m_{k}, \, k=1,\cdots, N
  }}  
  \varepsilon^{\bm{m}} (\mathcal{T}_{K} C(0))(t), 
  \label{N-param power series}
\end{align}
where we have used the fact that 
$C(0) = \sum_{\bm{n}}\varepsilon^{\bm{n}}C_{\bm{n}}(0) = 
\sum_{\bm{n}}\varepsilon^{\bm{n}}Q_{\bm{n}}(0)$. 
\end{Prop}
If the control is piecewise constant in $t$, the continuation along time 
of solutions of the form \eqref{N-param power series} 
with respective constant controls will yield a solution  
to \eqref{Schrodinger eq with 2-para}. 

\subsection{Controllability of the NMR systems}
From Eq.~\eqref{N param vect fields}, we observe that 
the control vector field are vertical, 
so that the measure defined in 
\eqref{measurement} does not change along this vector field. 
Hence, the control would not make the NMR system entangled 
without the coupling with the drift vector field. 
With this observation in mind, 
we gain insight into \eqref{N-param power series} 
to understand how the coupling occurs. 
We first look into the operators 
$\mathcal{T}_{k}$ defined in \eqref{eq: T_k}. 
Writing out the integrand of \eqref{eq: T_k} in the form of a power series 
in $s$, we obtain 
\begin{align*}
  (\mathcal{T}_{k}F)(t) 
  =& 
  -i\int_{0}^{t}\left\{
    \Sigma^{(k)}_{1}F(s)(\Sigma^{(k)}_{2})^{T} 
    -s\Big(
      [\xi_{1},\Sigma^{(k)}_{1}]F(s)(\Sigma^{(k)}_{2})^{T}
      +
      \Sigma^{(k)}_{1}F(s)[\xi_{2}, \Sigma^{(k)}_{2}]^{T}
    \Big) 
    \right.  \\
  & \quad \quad \quad \quad 
    +\frac{s^{2}}{2!}\Big(
      [\xi_{1},[\xi_{1}, \Sigma^{(k)}_{1}]]F(s)(\Sigma^{(k)}_{2})^{T}
      +2[\xi_{1},\Sigma^{(k)}_{1}]F(s)[\xi_{2}, \Sigma^{(k)}_{2}]^{T} 
       \\
  & \hspace{2.8cm} 
      \left. 
      +\Sigma^{(k)}_{1}F(s)[\xi_{2}, [\xi_{2}, \Sigma^{(k)}_{2}]]^{T}
    \Big)
    + \cdots 
    \right\}ds. 
\end{align*}
By applying Proposition \ref{Prop: comp of brackets}, 
we can put the above equation in the form 
\begin{align}
  & 
  (\mathcal{T}_{k}F)(t) \notag \\
  = & 
  \int_{0}^{t}\left\{ 
    -X_{i\Sigma^{(k)}_{1}\otimes\Sigma^{(k)}_{2}}(F(s)) 
    + 
    s\Big(
    [X_{\xi_{1}\otimes I}, -X_{i\Sigma^{(k)}_{1}\otimes\Sigma^{(k)}_{2}}](F(s))
    + 
    [X_{I\otimes\xi_{2}}, -X_{i\Sigma^{(k)}_{1}\otimes\Sigma^{(k)}_{2}}](F(s))
    \Big) 
  \right.  \notag \\
  & \quad \quad 
  + \frac{s^{2}}{2!} \Big(
      [X_{\xi_{1}\otimes I},
      [X_{\xi_{1}\otimes I}, -X_{i\Sigma^{(k)}_{1}\otimes
      \Sigma^{(k)}_{2}}]](F(s))
      +
      2[X_{\xi_{1}\otimes I},
      [X_{I\otimes\xi_{2}}, -X_{i\Sigma^{(k)}_{1}\otimes
      \Sigma^{(k)}_{2}}]](F(s)) \notag \\
  & \quad \quad \quad \quad \quad 
  \label{expanded Tk}
      \left. 
      + 
      [X_{I\otimes\xi_{2}}, 
      [X_{I\otimes\xi_{2}}, -X_{i\Sigma^{(k)}_{1}\otimes
      \Sigma^{(k)}_{2}}]](F(s)) 
      \,\Big)
      + \cdots 
      \right\}ds.
\end{align}
Since the expansion \eqref{expanded Tk} means that 
the drift and control vector fields are coupled indeed, 
and since the operator $\mathcal{T}_K$ is the composition of these 
$\mathcal{T}_k$, 
Eq.~\eqref{N-param power series} shows that 
the drift and control vector fields work together to give rise to 
the solution. 
If the coupling between the drift and control vector fields 
generates vector fields with non-vanishing 
horizontal components, the NMR system will get more entangled 
by controls. 

Incidentally, if the NMR system is controllable, 
then any initial state can be transferred into, say,  
a maximally entangled state by some controls. 
This means that the control vector field should be coupled with 
the drift vector field in the case of controllable NMR systems. 

According to \cite{AD'A03}, 
the system \eqref{Schrodinger equation 0} is said to be 
pure state controllable if, 
for any initial and final states, 
there are controls $\xi_{\alpha}, \alpha=1,\cdots, n$, 
and a finite time $T>0$ such that the initial state is 
transferred into the final one on the finite time $T$. 
Further, the system \eqref{Schrodinger equation 0} is 
pure state controllable if and only if 
the Lie algebra $\mathcal{L}_{\hat{H}}$ generated by 
$\{ -i\hat{H}_{d} \} \cup 
\{ i\sigma_{j}^{(\alpha,n)} \}_{j=1,2}^{\alpha = 1,\cdots, n}$ 
is isomorphic to $\frak{su}(2^{n})$. 
For the time being, we are interested in the coupling between 
the drift and control vector fields, 
and postpone the controllability problem to Appendix~B, 
in which we will give a necessary and
sufficient condition for controllability.  

However, it is of great help to apply the controllability theorem 
to the NMR system. 
To state the theorem, we need to introduce the notion of spin graph. 
Here, the spin graph associated with a given drift Hamiltonian 
$\hat{H}_{d}$ 
is defined to be the un-oriented graph that has $n$ nodes 
and edges joining the nodes labeled as $\alpha$ and $\beta$ such that 
$J_{\alpha\beta}\neq 0$, where nodes and edges represent spin-$\frac12$ 
particles and interactions among spin-$\frac{1}{2}$ particles such that 
$J_{\alpha\beta}\neq 0$, respectively. 
In terms of the spin graph, the controllability theorem is stated 
as follows (see Theorem B.8): 
The NMR system is controllable, if and only if the associated spin graph 
is connected. 
In view of this, we may set some of $J_{\alpha\beta}$ 
in \eqref{assoc vec fields} to be zero. 
We here take 
\begin{equation}
  \left\{ \begin{array}{ll}
    J_{1\beta}\neq 0, & \ell+1\leq \beta \leq \ell +m, \\ 
    J_{\alpha,\ell+1}\neq 0, &  1\leq \alpha \leq \ell, \\ 
    J_{\alpha\beta}=0, & {\rm otherwise}. 
  \end{array}\right. 
\end{equation}
Then, the associated spin graph is connected, so that the NMR system 
with these coupling constants is controllable. 
The drift vector field in \eqref{assoc vec fields} is now expressed as 
\begin{equation*}
  -iH_d(C) = 
  -i[\sigma^{(1,\ell)}_3]C
     \Bigl( 
       \sum_{\ell+1\leq \beta \leq \ell+m}J_{1\beta}
       [\sigma_3^{(\beta-\ell,m)}]^T 
     \Bigr)
  -i\Bigl( 
    \sum_{2\leq \alpha \leq \ell}J_{\alpha,\ell+1}
    [\sigma^{(\alpha,\ell)}_3]
  \Bigr) C [\sigma_3^{(1,m)}]^T . 
\end{equation*}
In this case, one has $N=2$, and 
 \begin{align*}
  & \Sigma_1^{(1)}=[\sigma_3^{(1,\ell)}], \quad 
    \Sigma_2^{(1)}=\sum_{\ell+1 \leq \beta \leq \ell+m}
                   J^{(1)}_{\beta}[\sigma_3^{(\beta-\ell),m)}]^T, 
    \quad \varepsilon_1 J_{\beta}^{(1)}=J_{1\beta}, \\
  & \Sigma_1^{(2)}=\sum_{2\leq \alpha \leq \ell}
                   J_{\alpha}^{(2)}[\sigma_3^{(\alpha,\ell)}], \quad 
    \Sigma_2^{(2)}=[\sigma_3^{(1,m)}]^T, \quad 
    \varepsilon_2 J_{\alpha}^{(2)}=J_{\alpha,\ell+1}, \\
  & \xi_1=\sum_{1\leq \alpha \leq \ell} [I^{\otimes (\alpha-1)} \otimes 
    \xi_{\alpha} \otimes I^{\otimes (\ell-\alpha)}], \quad 
    \xi_2=\sum_{\ell+1\leq \beta \leq \ell+m} [I^{\otimes (\beta-\ell-1)} 
    \otimes \xi_{\beta} \otimes I^{\otimes (m-\beta+\ell)}]^T, 
 \end{align*}
 where $\xi_{\alpha},\xi_{\beta}\in 
 {\rm span}_{\mathbb{R}}\{i\sigma_1,i\sigma_2\}$. 

\section{Controls in two-qubit NMR systems}
\label{Control problem in two-qubit systems}

In this section, we deal with the two-qubit case. 
Our system is expressed as 
\begin{equation}
  \frac{dC}{dt}
  = 
  -JX_{i\sigma_{3}\otimes\sigma_{3}} 
  + X_{\xi_{1}\otimes I} + X_{I\otimes\xi_{2}} 
  = 
  -iJ\sigma_{3}C\sigma_{3}^{T} + \xi_{1}C + C\xi_{2}^{T}, 
  \label{NMR model}
\end{equation}
where $J>0$ is the coupling constant between the two qubits, 
and where $\xi_{1}, \xi_{2}$ are constant controls taking values in 
$\Span_{\mathbb{R}}\{ i\sigma_{1}, i\sigma_{2} \}$.
In comparison to \eqref{Schrodinger eq with 2-para}, 
this equation is easier to solve, and we can investigate solutions 
in detail to understand how controls make the NMR system more entangled. 

\subsection{Two-qubit NMR systems without control}
To understand why one needs controls to make the NMR systems entangled, 
we look into the uncontrolled system with 
$\xi_{1} = \xi_{2} = 0$ in \eqref{NMR model}. 
If $\xi_{1} = \xi_{2} = 0$, Eq.~\eqref{NMR model} is easily integrated 
to give the solution 
\begin{equation}
  C(t) 
  = 
  \varphi_{-i\sigma_{3}\otimes\sigma_{3}}^{Jt}(C(0))
  = 
  \begin{pmatrix}
    e^{-iJt/4}c_{00}(0) & e^{iJt/4}c_{01}(0) \\
    e^{iJt/4}c_{10}(0) & e^{-iJt/4}c_{11}(0) 
  \end{pmatrix}. 
  \label{sol: uncontrolled}
\end{equation}
Then the measure $F(C(t))=\det(C(t)^*C(t))$ is evaluated as  
\begin{multline*} 
  F(C(t)) = 
  |c_{00}(0)c_{11}(0)|^{2}+|c_{01}(0)c_{10}(0)|^{2} \\
  - 
  \left(
    e^{-iJt}c_{00}(0)\overline{c_{01}(0)}\,\overline{c_{10}(0)}c_{11}(0))
    +
    e^{iJt}\overline{c_{00}(0)}c_{01}(0)c_{10}(0)\overline{c_{11}(0))}
  \right), 
\end{multline*}
which gives rise to the inequality 
\begin{equation} 
  F(C(t)) \leq 
  (|c_{00}(0)c_{11}(0)|+|c_{01}(0)c_{01}(0)|)^{2}
  =: D(C(0)). 
  \label{estimate of F}
\end{equation}
Eq.\! \eqref{estimate of F} implies that 
the quantity $D(C(0))$ determines whether 
the state driven by \eqref{NMR model} with the initial state $C(0)$ reaches 
a maximally entangled state without any controls or not. 

Let us be reminded of the fact that the measure $F$ ranges over 
$0\leq F(C) \leq 1/4$ for the two-qubit, where 
the minimum or the maximum occurs, according to whether $C$ is separable or 
maximally entangled. 
If $D(C(0))=1/4$, the state $C(t)$ evolves to reach 
a maximally entangled state without control even if $C(0)$ is separable. 
For instance, starting with the separable state 
\begin{equation} 
  C(0) = 
  \frac{1}{2}
  \begin{pmatrix} 1 & 1 \\ 1 & 1 \end{pmatrix}, 
  \label{example of C(0): 1}
\end{equation}
the solution curve $C(t)$ passes through the maximally entangled state 
at $t=\pi/J$, 
\begin{equation*}
  C\left( \frac{\pi}{J} \right) = 
  \frac{1}{2}
  \begin{pmatrix} 
    e^{-i\pi/4} & e^{i\pi/4} \\
    e^{i\pi/4} & e^{-i\pi/4} 
  \end{pmatrix}. 
\end{equation*}
We note in addition that $F(C(t))=(1-\cos Jt)/8$ in this case. 
In contrast with this, if $D(C(0))=0$, 
then $C(0)$ is separable, and the state $C(t)$ is always separable 
because of $0 \leq F(C(t)) \leq D(C(0)) = 0$. 
An example of such an initial state is 
\begin{equation} 
  C(0) = 
  \diag{e^{i\theta}, 0}. 
  \label{example of C(0): 2}
\end{equation}
This initial state never gets entangled without controls. 
From these facts, we observe that the separable states 
$\frac12(|0\rangle +|1\rangle) \otimes (|0\rangle +|1\rangle)$ and 
$e^{i\theta}|0\rangle\otimes |0\rangle$ are of different 
nature with respect to the drift Hamiltonian.

\subsection{Behaviours of solutions in the two-qubit NMR system}
We turn to the controlled system \eqref{NMR model} with 
$\xi_1$ and $\xi_2$ non-vanishing constants. 
Since the associated spin graph is connected, 
the two-qubit NMR system is controllable, 
so that any separable state can get entangled. 
This implies that the coupling between the drift and control 
vector fields must occur. 

We solve the Schr\"odinger equation \eqref{NMR model} as 
a power series of the coupling constant $J>0$. 
As in the previous section, we assume that 
the $J$ is small enough, and that 
the solution $C(t)$ of \eqref{NMR model} can be expanded
into a power series in $J$ as $C(t)=\sum_{n=0}^{\infty}J^{n}C_{n}(t)$. 
Substituting this into \eqref{NMR model} brings about 
the series of differential equations for $C_{n}$, 
\begin{equation}
  \left\{
  \begin{aligned}
    \dot{C}_{0} 
    &= 
    \quad \quad \quad \quad \quad \quad \,\,\, 
    \xi_{1}C_{0} + C_{0}\xi_{2}^{T}, \\
    \dot{C}_{n} 
    &=  
    -i\sigma_{3}C_{n-1}\sigma_{3}^{T} 
    +\xi_{1}C_{n} + C_{n}\xi_{2}^{T}, 
    \quad 
    n \geq 1. 
  \end{aligned}
  \right.
  \label{eq: expansion}
\end{equation}
These equations can be integrated inductively, 
like \eqref{two-param eq: expansion}. 
In fact, from the above equations, differential equations for 
$e^{-t\xi_{1}}C_{n}(t)e^{-t\xi_{2}^{T}}$ 
are easily obtained  and integrated to give 
\begin{align}
  & 
  e^{-t\xi_{1}} C(t) e^{-t\xi_{2}^{T}}
  =
  \sum_{n=0}^{\infty} 
  J^{n} e^{-t\xi_{1}} C_{n}(t) e^{-t\xi_{2}^{T}} \notag\\
  =& \,
  \sum_{n=0}^{\infty} 
  (-iJ)^{n}
  \int_{0}^{t} \!\! ds_{n} \, 
  \Ad_{e^{-s_{n}\xi_{1}}}(\sigma_{3}) 
  \int_{0}^{s_{n}} \!\!\!\! ds_{n-1} \, 
  \Ad_{e^{-s_{n-1}\xi_{1}}}(\sigma_{3}) 
  \cdots 
  \int_{0}^{s_{2}} \!\! ds_{1} \, 
  \Ad_{e^{-s_{1}\xi_{1}}}(\sigma_{3}) \notag\\
  & \quad \quad \quad \quad \quad \quad \quad \quad 
  \times 
  C(0) \, 
  \Ad_{e^{s_{1}\xi_{2}^{T}}}(\sigma_{3}^{T}) 
  \cdots 
  \Ad_{e^{s_{n-1}\xi_{2}^{T}}}(\sigma_{3}^{T}) 
  \Ad_{e^{s_{n}\xi_{2}^{T}}}(\sigma_{3}^{T}), 
  \label{power series of J}
\end{align}
where we have used the assumption that 
$C(0) = \sum_{n=0}^{\infty}J^{n}C_{n}(0)$. 
The following proposition is easy to prove. 
\begin{Prop} 
The power series \eqref{power series of J} uniformly converges to 
a function $P(t)$, with which the solution of 
the NMR system \eqref{NMR model} is 
given by $C(t) = e^{t\xi_{1}}P(t)e^{t\xi_{2}^{T}}$. 
\end{Prop}

So far we have found the solution as the power series 
\eqref{power series of J} in the coupling constant $J$. 
We now analyze the solution in order to see in detail that 
the coupling between the drift and control vectors actually occurs 
to make the NMR system more entangled. 
We denote the $n$-th order term in $J$ of 
\eqref{power series of J} by $P_{n}(t)$. 
Then, the solution is expressed as 
$C(t) = \sum_{n=0}^{\infty} J^{n} e^{t\xi_{1}}P_{n}(t) e^{t\xi_{2}^{T}}$. 
This explains how the controls $\xi_{1}$ and $\xi_{2}$ contribute to 
entanglement promotion of the system. 
The zeroth-order term does not 
make the system more entangled because of 
$F(e^{t\xi_{1}}P_{0}(0)e^{t\xi_{2}^{T}}) = F(C(0))$. 
Contrarily, the higher-order terms in $J$ are expected to 
make the NMR system entangled in general. 
We look into \eqref{power series of J} in detail. 
The first-order term $P_{1}(t)$ of \eqref{power series of J} is 
put in the form $(\mathcal{T}_{1}C(0))(t)$. Then, 
as is seen from \eqref{expanded Tk}, $P_1(t)$ is expressed as 
\begin{align}
  P_{1}(t) 
  = & 
  -tX_{i\sigma_{3}\otimes\sigma_{3}}(C(0)) 
  + 
  \frac{t^{2}}{2!}\Big(
    [X_{\xi_{1}\otimes I}, -X_{i\sigma_{3}\otimes\sigma_{3}}]
    + 
    [X_{I\otimes\xi_{2}}, -X_{i\sigma_{3}\otimes\sigma_{3}}]
  \Big)\Big|_{C(0)} \notag\\
  & 
  + \frac{t^{3}}{3!} \Big(
      [X_{\xi_{1}\otimes I},
      [X_{\xi_{1}\otimes I}, -X_{i\sigma_{3}\otimes\sigma_{3}}]]
      +
      [X_{\xi_{1}\otimes I},
      [X_{I\otimes\xi_{2}}, -X_{i\sigma_{3}\otimes\sigma_{3}}]] \notag\\
  & \quad \quad 
      + 
      [X_{I\otimes\xi_{2}}, 
      [X_{\xi_{1}\otimes I}, -X_{i\sigma_{3}\otimes\sigma_{3}}]]
      + 
      [X_{I\otimes\xi_{2}}, 
      [X_{I\otimes\xi_{2}}, -X_{i\sigma_{3}\otimes\sigma_{3}}]]
    \,\Big)\Big|_{C(0)}
    + \cdots .
  \label{first-order term}
\end{align}
This shows that the drift and control vector fields are coupled indeed. 
We note that 
$[X_{I\otimes\xi_{2}}, $
$[X_{\xi_{1}\otimes I}, -X_{i\sigma_{3}\otimes\sigma_{3}}]]=
[X_{\xi_{1}\otimes I}, 
[X_{I\otimes \xi_{2}}, -X_{i\sigma_{3}\otimes\sigma_{3}}]]$ 
by the Jacobi identity and that 
$[X_{\xi_{1}\otimes I}$, $X_{I\otimes\xi_{2}}]=0$. 

By a similar computation, the second-order term $P_{2}(t)$ of 
\eqref{power series of J} is written out as 
\begin{align*} 
  & P_{2}(t) = (\mathcal{T}_1\circ \mathcal{T}_1C(0))(t) \\
  =& 
  \frac{t^{2}}{2!} \,
  (-i\sigma_{3}(-i\sigma_{3}C(0)\sigma_{3}^{T})\sigma_{3}^{T}) 
  -i \frac{t^{3}}{3!} \, 
  \sigma_{3}\Big(
    [\xi_{1},\sigma_{3}]C(0)\sigma_{3}^{T}
    +
    \sigma_{3}C(0)[\xi_{2},\sigma_{3}]^{T}
  \Big)\sigma_{3}^{T} 
  + \cdots. 
\end{align*}
We formally denote the first and second terms of 
the right-hand side of the above equation by 
$\frac{t^{2}}{2!}(-X_{i\sigma_{3}\otimes\sigma_{3}})^{2}|_{C(0)}$ and 
$\frac{t^{3}}{3!}(-X_{i\sigma_{3}\otimes\sigma_{3}})
([X_{\xi_{1}\otimes I}, -X_{i\sigma_{3}\otimes\sigma_{3}}]+
[X_{I\otimes\xi_{2}}, -X_{i\sigma_{3}\otimes\sigma_{3}}])|_{C(0)}$, 
respectively. 
The first term 
$\frac{t^{2}}{2!}(-X_{i\sigma_{3}\otimes\sigma_{3}})^{2}$ 
comes from the Taylor expansion of the flow 
$\varphi_{-i\sigma_{3}\otimes\sigma_{3}}^{t}$ 
with respect to $t$. 

It is inductively shown that 
the $n$-th order term $P_{n}(t)$ is expanded into a power series of $t$ 
with the lowest-order term 
$\frac{t^{n}}{n!}(-X_{i\sigma_{3}\otimes\sigma_{3}})^{n}$, 
which comes from the Taylor expansion 
$\varphi_{-i\sigma_{3}\otimes\sigma_{3}}^{t}(C(0)) = 
\sum_{n=0}^{\infty}
\frac{t^{n}}{n!}(-X_{i\sigma_{3}\otimes\sigma_{3}})^{n}|_{C(0)}$. 
Hence, we have 

\begin{Prop}
If $t$ is sufficiently small, 
the solution $C(t)$ of \eqref{NMR model} allows of the approximation of 
the form 
\begin{align}
  &
  e^{-t\xi_{1}}C(t)e^{-t\xi_{2}^{T}} 
  - 
  \varphi_{-i\sigma_{3}\otimes\sigma_{3}}^{Jt}(C(0)) \notag\\
  \sim & \, 
  \frac{Jt^{2}}{2!}
  \Big(
    [X_{\xi_{1}\otimes I}, -X_{i\sigma_{3}\otimes\sigma_{3}}]
    + 
    [X_{I\otimes\xi_{2}}, -X_{i\sigma_{3}\otimes\sigma_{3}}]
  \Big)\Big|_{C(0)} \notag\\
  & 
  + \frac{J^{2}t^{3}}{3!} (-X_{i\sigma_{3}\otimes\sigma_{3}})
  \Big(  
    [X_{\xi_{1}\otimes I}, -X_{i\sigma_{3}\otimes\sigma_{3}}]
    + 
    [X_{I\otimes\xi_{2}}, -X_{i\sigma_{3}\otimes\sigma_{3}}]    
  \Big) \biggr|_{C(0)} \notag\\
  & 
  + \frac{Jt^{3}}{3!} \Big(
      [X_{\xi_{1}\otimes I},
      [X_{\xi_{1}\otimes I}, -X_{i\sigma_{3}\otimes\sigma_{3}}]]
      +
      [X_{\xi_{1}\otimes I},
      [X_{I\otimes\xi_{2}}, -X_{i\sigma_{3}\otimes\sigma_{3}}]] \notag\\
  & \quad \quad \quad 
      + 
      [X_{I\otimes\xi_{2}}, 
      [X_{\xi_{1}\otimes I}, -X_{i\sigma_{3}\otimes\sigma_{3}}]]
      + 
      [X_{I\otimes\xi_{2}}, 
      [X_{I\otimes\xi_{2}}, -X_{i\sigma_{3}\otimes\sigma_{3}}]]
  \,\Big)\Big|_{C(0)} 
  + \cdots,  
  \label{approximation of C}
\end{align}
where $\varphi_{-i\sigma_{3}\otimes\sigma_{3}}^{Jt}(C(0))$ is 
the solution of the uncontrolled system 
$dC/dt = -JX_{i\sigma_{3}\otimes\sigma_{3}}$ 
with the initial state $C(0)$. 
\end{Prop}

\subsection{Bases of vertical and horizontal subspaces}
So far we have shown that the drift and control vector fields 
are coupled. We now wish to show that the coupling generates vector fields 
with non-vanishing horizontal components. 
To this end, we have to look into vertical and horizontal vector fields 
in detail. 

Since the subspaces $V_C$ and $H_C$ depend on the stratum to which 
$C$ belongs, 
we make a brief review of a stratification of the state space 
$M$ for the two-qubit system after \cite{Iw07-2}. 
According to singular values $\lambda_{1}(C) \geq \lambda_{2}(C)$ 
of $C \in M$, 
the state space $M$ of the two-qubit system is stratified into 
three $G$-invariant strata as $M = M_{0} \sqcup M_{1} \sqcup M_{2}$, 
where 
\begin{equation}
  \begin{aligned}
    M_{0}
    &= 
    \{ C \in M \,|\, \lambda_{1}(C) = 1, \lambda_{2}(C) = 0 \}, \\
    M_{1}
    &= 
    \{ C \in M \,|\, \lambda_{1}(C) > \lambda_{2}(C) > 0 \}, \\
    M_{2} 
    &= 
    \{ C \in M \,|\, \lambda_{1}(C) = \lambda_{2}(C) = 1/\sqrt{2} \}. 
  \end{aligned}
  \label{strata}
\end{equation}
The strata $M_{0}$ and $M_{2}$ are 
the sets of separable and maximally entangled states, respectively. 

We now deals with bases of the subspaces $V_{C}$ and $H_{C}$. 
We take up a typical matrix $\Lambda=\diag{\lambda_{1},\lambda_{2}}$ 
with $\lambda_{1} \geq \lambda_{2}$, 
to which all the matrices with the same singular values are translated 
by the $\U{2}\otimes \U{2}$ action. 
Let us start with the vertical subspace $V_{\Lambda}$. 
Proposition \ref{Prop: isotropy subgrp and orbit} with \eqref{strata}
implies that 
\begin{equation}
  \dim V_{\Lambda} = 
  \begin{cases}
    \,\, 6, & \text{ if } \Lambda \in M_{1}, \\
    \,\, 5, & \text{ if } \Lambda \in M_{0}, \\
    \,\, 4, & \text{ if } \Lambda \in M_{2}. 
  \end{cases}
  \label{dim of V}
\end{equation}\
Tangent vectors $X_{iI\otimes I}(\Lambda), 
X_{i\sigma_{j}\otimes I}(\Lambda)$ and 
$X_{iI\otimes\sigma_{j}}(\Lambda)$ for 
$j=1,2,3$, defined in \eqref{fundamental vector}, 
span the tangent space $V_{\Lambda}$. 
However, some of these vectors coincide with one another,  

\vspace{1mm}
\noindent 
(1) 
$X_{i\sigma_{3}\otimes I}(\Lambda) 
= 
X_{iI\otimes\sigma_{3}}(\Lambda)$ 
at any diagonal matrix $\Lambda$, 

\noindent 
(2) 
$X_{i\sigma_{3}\otimes I}(\Lambda) 
= 
X_{iI\otimes\sigma_{3}}(\Lambda) 
= 
X_{iI\otimes I}(\Lambda)$ 
at $\Lambda = \diag{1,0} \in M_{0}$, 

\noindent 
(3) 
$X_{i\sigma_{1}\otimes I}(\Lambda) 
= 
X_{iI\otimes\sigma_{1}}(\Lambda)$ 
and 
$X_{i\sigma_{2}\otimes I}(\Lambda) 
= 
-X_{iI\otimes\sigma_{2}}(\Lambda)$, 
at $\Lambda = I/\sqrt{2} \in M_{2}$. 

\vspace{1mm}
\noindent 
Thus, respective bases of $V_{\Lambda}$ are obtained as follows: 

\begin{Prop}
\label{Prop: basis of V}
Let $\Lambda=\diag{\lambda_1,\lambda_2}$. \\
\noindent 
(i) If $\Lambda$ lies in the principal stratum $M_{1}$, then  
$V_{\Lambda}$ has the basis 
\begin{equation*}
  X_{iI\otimes I}(\Lambda),\, 
  X_{i\sigma_{3}\otimes I}(\Lambda)
  =
  X_{iI\otimes \sigma_{3}}(\Lambda),\, 
  X_{i\sigma_{1}\otimes I}(\Lambda),\, 
  X_{iI\otimes \sigma_{1}}(\Lambda),\, 
  X_{i\sigma_{2}\otimes I}(\Lambda),\, 
  X_{iI\otimes \sigma_{2}}(\Lambda). 
\end{equation*}

\noindent 
(ii) If $\Lambda$ is a separable state, {\it i.e.}, 
$\Lambda=\diag{1,0} \in M_{0}$, 
then $V_{\Lambda}$ has the basis 
\begin{equation*}
  X_{iI\otimes I}(\Lambda)
  =
  X_{i\sigma_{3}\otimes I}(\Lambda)
  =
  X_{iI\otimes \sigma_{3}}(\Lambda),\,
  X_{i\sigma_{1}\otimes I}(\Lambda),\, 
  X_{iI\otimes \sigma_{1}}(\Lambda),\, 
  X_{i\sigma_{2}\otimes I}(\Lambda),\, 
  X_{iI\otimes \sigma_{2}}(\Lambda). 
\end{equation*}

\noindent 
(iii) If $\Lambda$ is maximally entangled, {\it i.e.}, 
$\Lambda=I/\sqrt{2} \in M_{2}$, 
then $V_{\Lambda}$ has the basis 
\begin{equation*}
  X_{iI\otimes I}(\Lambda),\, 
  X_{i\sigma_{3}\otimes I}(\Lambda)
  =
  X_{iI\otimes \sigma_{3}}(\Lambda),\, 
  X_{i\sigma_{1}\otimes I}(\Lambda)
  =
  X_{iI\otimes \sigma_{1}}(\Lambda),\, 
  X_{i\sigma_{2}\otimes I}(\Lambda)
  =
  -X_{iI\otimes\sigma_{2}}(\Lambda). 
\end{equation*}
Since the $G$-action \eqref{G-action} is isometric, 
the singular value decomposition \eqref{singular value decomposition}, 
viewed as a map, $\Lambda \mapsto C = (g\otimes h)\cdot\Lambda$,  
gives rise to a basis of $V_{C}$ from that of $V_{\Lambda}$ 
by the differential map $(g\otimes h)_*$. 
\end{Prop}

The next task is to find a basis of 
the horizontal subspace $H_{C}$. 
From \eqref{dim of V}, the dimension of $H_{C}$ proves to be 
\begin{equation}
\label{dim of H}
  \dim H_{C} = 
  \begin{cases}
    1, & \text{ if } C \in M_{1}, \\
    2, & \text{ if } C \in M_{0}, \\
    3, & \text{ if } C \in M_{2}. 
  \end{cases}
\end{equation}
The following proposition is easily verified 
by a straightforward computation. 

\begin{Prop}
\label{Prop: basis of H}
Let $\Lambda=\diag{\lambda_1,\lambda_2}$.\\
\noindent 
(i) If $\Lambda$ lies in the principal stratum $M_{1}$, 
then the horizontal subspace $H_{\Lambda}$ is 
a one-dimensional vector space spanned by 
\begin{equation*} 
  X_{i\sigma_{1}\otimes\sigma_{2}}(\Lambda)
  = X_{i\sigma_{2}\otimes\sigma_{1}}(\Lambda). 
\end{equation*}
 
\noindent 
(ii) If $\Lambda$ is separable, {\it i.e.}, $\Lambda\in M_{0}$, 
or $\Lambda = \diag{1,0}$, then $H_{\Lambda}$ has the basis
\begin{equation*}
    X_{i\sigma_{1}\otimes\sigma_{1}}(\Lambda),\, 
    X_{i\sigma_{1}\otimes\sigma_{2}}(\Lambda)
    = X_{i\sigma_{2}\otimes\sigma_{1}}(\Lambda). 
\end{equation*}

\noindent 
(iii) If $\Lambda$ is maximally entangled, {\it i.e.}, $\Lambda\in M_{2}$, 
or $\Lambda = I/\sqrt{2}$, then $H_{\Lambda}$ has the basis 
\begin{equation*}
    iX_{i\sigma_{1}\otimes I}(\Lambda), \, 
    iX_{i\sigma_{2}\otimes I}(\Lambda), \, 
    iX_{i\sigma_{3}\otimes I}(\Lambda).
\end{equation*}
A basis of $V_{C}$ are formed from that of $V_{\Lambda}$ 
by the differential map $(g\otimes h)_*$.
\end{Prop}

From these propositions, we observe that 
when $\Lambda=I/\sqrt{2} \in M_{2} \cong \U{2}$, 
the vertical subspace $V_{\Lambda}$ is identified with  
the space $\frak{u}(2)$ of $2\times 2$ skew Hermitian matrices, 
and that the horizontal subspace $H_{\Lambda}$ with $\Lambda = I/\sqrt{2}$ 
is identified with the space $i\frak{su}(2)$ of traceless Hermitian matrices. 
Further, the basis vectors of $H_{\Lambda}$ given in 
Proposition \ref{Prop: basis of H} (iii) are 
alternatively expressed as 
\begin{align}
  & 
  iX_{i\sigma_{1}\otimes I}(\frac{1}{\sqrt{2}}I) 
  = 
  iX_{iI\otimes\sigma_{1}}(\frac{1}{\sqrt{2}}I) 
  = 
  2X_{i\sigma_{3}\otimes\sigma_{2}}(\frac{1}{\sqrt{2}}I)
  = 
  2X_{i\sigma_{2}\otimes\sigma_{3}}(\frac{1}{\sqrt{2}}I),  \nonumber \\
  & 
  iX_{i\sigma_{2}\otimes I}(\frac{1}{\sqrt{2}}I) 
  = 
  -iX_{iI\otimes\sigma_{2}}(\frac{1}{\sqrt{2}}I)
  = 
  2X_{i\sigma_{1}\otimes\sigma_{3}}(\frac{1}{\sqrt{2}}I)
  = 
  -2X_{i\sigma_{3}\otimes\sigma_{1}}(\frac{1}{\sqrt{2}}I), 
  \label{HorVecM2}   \\
  &
  iX_{i\sigma_{3}\otimes I}(\frac{1}{\sqrt{2}}I) 
  = 
  iX_{iI\otimes\sigma_{3}}(\frac{1}{\sqrt{2}}I) 
  = 
  2X_{i\sigma_{2}\otimes\sigma_{1}}(\frac{1}{\sqrt{2}}I)
  = 
  2X_{i\sigma_{1}\otimes\sigma_{2}}(\frac{1}{\sqrt{2}}I).  \nonumber
\end{align}

\begin{Rem}
The vector fields 
$X_{i\sigma_{2}\otimes\sigma_{1}}$ and 
$X_{i\sigma_{1}\otimes\sigma_{1}}$ 
are generated, by taking the Lie brackets among 
$X_{i\sigma_{3}\otimes\sigma_{3}}, X_{i\sigma_{1}\otimes I}$ and 
$X_{iI\otimes\sigma_{2}}$, 
\begin{equation}
  X_{i\sigma_{2}\otimes\sigma_{1}} 
  = 
  [X_{iI\otimes\sigma_{2}}, 
  [X_{i\sigma_{1}\otimes I}, X_{i\sigma_{3}\otimes\sigma_{3}}]],  
  \quad 
  X_{i\sigma_{1}\otimes\sigma_{1}} 
  = 
  -[X_{iI\otimes\sigma_{2}}, 
  [X_{iI\otimes\sigma_{2}}, X_{i\sigma_{3}\otimes\sigma_{3}}]],  
\end{equation}
respectively. 
The first and second equations of the above imply that, 
if doubly coupled, the drift and control vector fields generates 
a horizontal vector at $C=\Lambda\in M_1$ and 
another horizontal vector at $C=\Lambda\in M_0$, respectively. 
\end{Rem}

\subsection{The concurrence and control} 
We now investigate the right-hand side of 
\eqref{approximation of C} to examine the effect of the control on 
entanglement promotion. 
We suppose that 
the initial state $C(0)$ is a diagonal matrix 
$\Lambda=\diag{\lambda_{1},\lambda_{2}}$.  

The terms in the $\mathcal{O}(Jt^{2})$-term of the right-hand side of 
\eqref{approximation of C} are expanded as 
\begin{equation}
  \begin{array}{ll}
    & 
    [X_{i\xi_{1}\otimes I}, -X_{i\sigma_{3}\otimes\sigma_{3}}] 
    = 
    -x_{1}X_{i\sigma_{2}\otimes\sigma_{3}} 
    +y_{1}X_{i\sigma_{1}\otimes\sigma_{3}}, \\
    & 
    [X_{iI\otimes\xi_{2}}, -X_{i\sigma_{3}\otimes\sigma_{3}}] 
    = 
    -x_{2}X_{i\sigma_{3}\otimes\sigma_{2}} 
    +y_{2}X_{i\sigma_{3}\otimes\sigma_{1}},  
  \end{array}
  \label{second order terms}
\end{equation}
where $\xi_{\alpha} = ix_{\alpha}\sigma_{1} + iy_{\alpha}\sigma_{2}, 
\alpha=1,2,$ 
with $x_{\alpha}, y_{\alpha} \in \mathbb{R}$. 
These vector fields are all vertical at $\Lambda$ 
if $\Lambda \ne I/\sqrt{2}$. 
Indeed, the tangent vectors 
$X_{i\sigma_{2}\otimes\sigma_{3}}(\Lambda), X_{i\sigma_{1}\otimes\sigma_{3}}(\Lambda), 
X_{i\sigma_{3}\otimes\sigma_{2}}(\Lambda)$, 
and $X_{i\sigma_{3}\otimes\sigma_{1}}(\Lambda)$ are 
written as linear combinations of the vertical vectors 
$X_{i\sigma_{1}\otimes I}(\Lambda)$, $X_{i\sigma_{2}\otimes I}(\Lambda)$, 
$X_{iI\otimes\sigma_{1}}(\Lambda)$ and $X_{iI\otimes\sigma_{2}}(\Lambda)$. 
For instance, one has 
\begin{subequations}
\label{sigma23 13}
\begin{align}
  X_{i\sigma_{2}\otimes\sigma_{3}}(\Lambda) 
  &= 
  \frac{\lambda_{1}^{2}+\lambda_{2}^{2}}{2(\lambda_{1}^{2}-\lambda_{2}^{2})}
  X_{i\sigma_{2}\otimes I}(\Lambda)
  + 
  \frac{\lambda_{1}\lambda_{2}}{\lambda_{1}^{2}-\lambda_{2}^{2}}
  X_{iI\otimes\sigma_{2}}(\Lambda), \\
  X_{i\sigma_{1}\otimes\sigma_{3}}(\Lambda) 
  &= 
  \frac{\lambda_{1}^{2}+\lambda_{2}^{2}}{2(\lambda_{1}^{2}-\lambda_{2}^{2})}
  X_{i\sigma_{1}\otimes I}(\Lambda) 
  -
  \frac{\lambda_{1}\lambda_{2}}{\lambda_{1}^{2}-\lambda_{2}^{2}}
  X_{iI\otimes\sigma_{1}}(\Lambda), 
\end{align}
\end{subequations}
if $\Lambda \ne I/\sqrt{2}$. 
If $\Lambda = I/\sqrt{2}$, they are found to be horizontal 
on account of \eqref{HorVecM2}. 

The $\mathcal{O}(J^{2}t^{3})$-term is also vertical at $\Lambda$, 
as is shown by 
\begin{equation*}
  -X_{i\sigma_{3}\otimes\sigma_{3}}
  [X_{\xi_{1}\otimes I}, -X_{i\sigma_{3}\otimes\sigma_{3}}]
  = 
  -\frac{1}{8} X_{\xi_{1}\otimes I}, 
  \quad 
  -X_{i\sigma_{3}\otimes\sigma_{3}}
  [X_{I\otimes\xi_{2}}, -X_{i\sigma_{3}\otimes\sigma_{3}}]
  = 
  -\frac{1}{8} X_{I\otimes\xi_{2}}. 
\end{equation*}

The $\mathcal{O}(Jt^{3})$-term in \eqref{approximation of C} is 
significant for entanglement promotion. 
The vector fields 
$[X_{\xi_{1}\otimes I}, [X_{\xi_{1}\otimes I}, 
  -X_{i\sigma_{3}\otimes\sigma_{3}}]]$,\,\,  
$[X_{I\otimes\xi_{2}}, [X_{I\otimes\xi_{2}}, 
  -X_{i\sigma_{3}\otimes\sigma_{3}}]]$ and 
$[X_{\xi_{1}\otimes I}, [X_{I\otimes\xi_{2}}, 
  -X_{i\sigma_{3}\otimes\sigma_{3}}]]$
are expressed as 
\begin{gather*}
  [X_{\xi_{1}\otimes I}, 
  [X_{\xi_{1}\otimes I}, -X_{i\sigma_{3}\otimes\sigma_{3}}]] 
  =  
  (x_{1}^{2}+y_{1}^{2})X_{i\sigma_{3}\otimes\sigma_{3}}, \\
  [X_{I\otimes\xi_{2}}, 
  [X_{I\otimes\xi_{2}}, -X_{i\sigma_{3}\otimes\sigma_{3}}]]
  =  
  (x_{2}^{2}+y_{2}^{2})X_{i\sigma_{3}\otimes\sigma_{3}}, \\
  [X_{\xi_{1}\otimes I}, 
  [X_{I\otimes\xi_{2}}, -X_{i\sigma_{3}\otimes\sigma_{3}}]] 
  = 
  -x_{1}x_{2}X_{i\sigma_{2}\otimes\sigma_{2}}
  +x_{1}y_{2}X_{i\sigma_{2}\otimes\sigma_{1}} 
  +y_{1}x_{2}X_{i\sigma_{1}\otimes\sigma_{2}}
    -y_{1}y_{2}X_{i\sigma_{1}\otimes\sigma_{1}}, 
\end{gather*}
respectively. 
Since $C(0)=\Lambda$ is diagonal, 
one has 
$X_{i\sigma_{3}\otimes\sigma_{3}}(\Lambda) = X_{iI\otimes I}(\Lambda) \in V_{\Lambda}$, 
so that the tangent vectors 
$[X_{\xi_{1}\otimes I}, [X_{\xi_{1}\otimes I}, -X_{i\sigma_{3}\otimes\sigma_{3}}]]\big|_{\Lambda}$ 
and 
$[X_{I\otimes\xi_{2}}, [X_{I\otimes\xi_{2}}, -X_{i\sigma_{3}\otimes\sigma_{3}}]]\big|_{\Lambda}$ 
are vertical. 
In contrast with this, from Proposition \ref{Prop: basis of H} and Remark 
after it, it turns out that the tangent vector 
$[X_{\xi_{1}\otimes I}, [X_{I\otimes\xi_{2}}, -X_{i\sigma_{3}\otimes\sigma_{3}}]]\big|_{\Lambda}$
contains the horizontal components 
$X_{i\sigma_{2}\otimes\sigma_{1}}(\Lambda), 
X_{i\sigma_{1}\otimes\sigma_{2}}(\Lambda)$. 
In particular, if $\xi_{1}$ and $\xi_{2}$ are taken as 
$\xi_{1}=i\sigma_{1}$ and $\xi_{2}=i\sigma_{2}$, respectively, 
the tangent vector in question is horizontal, 
\begin{equation*}
  [X_{i\sigma_{1}\otimes I}, 
  [X_{iI\otimes\sigma_{2}}, -X_{i\sigma_{3}\otimes\sigma_{3}}]]
  \big|_{\Lambda} 
  = 
  X_{i\sigma_{2}\otimes\sigma_{1}}(\Lambda) 
  \in 
  H_{\Lambda}. 
\end{equation*}
This shows that the drift vector field 
$X_{i\sigma_{3}\otimes\sigma_{3}}$ and the control vector fields 
$X_{\xi_{1}\otimes I}$, $X_{I\otimes\xi_{2}}$ are coupled to generate 
a horizontal vector at $\Lambda$.  
We note further that the horizontal vector emerges at the third-order term 
in $t$, if $\Lambda \ne I/\sqrt{2}$. 

We wish to evaluate the entanglement to confirm that 
the two-qubit system gets more entangled. 
As our entanglement measure and the concurrence are equivalent 
for the two-qubit system on account of $\sqrt{F(C(t))} = |\det C(t)|$, 
we here use the concurrence as a measure. 
From Eq.\! \eqref{approximation of C}, 
the concurrence is approximately evaluated as 
\begin{equation}
  \sqrt{F(C(t))} 
  \sim 
  \sqrt{F(\varphi_{-i\sigma_{3}\otimes\sigma_{3}}^{Jt}(\Lambda))} 
  +
  \frac{Jt^{3}}{3!\sqrt{F(\Lambda)}} 
  (dF)_{\Lambda}\Big(
    [X_{\xi_{1}\otimes I},
    [X_{I\otimes\xi_{2}}, -X_{i\sigma_{3}\otimes\sigma_{3}}]] 
  \Big), 
  \label{approximation of concurrence} 
\end{equation}
where we have used the following facts; (i) 
$F(e^{-t\xi_{1}}C(t)e^{-t\xi_{2}^{T}}) = F(C(t))$, 
(ii) 
$[X_{\xi_{1}\otimes I}, -X_{i\sigma_{3}\otimes\sigma_{3}}]$ and 
$[X_{I\otimes\xi_{2}}, -X_{i\sigma_{3}\otimes\sigma_{3}}]$ 
are vertical at $\Lambda$ 
if $\Lambda$ is not the maximally entangled state $I/\sqrt{2}$, 
(iii) 
$-X_{i\sigma_3\otimes \sigma_3} [X_{\xi_1\otimes I}, 
  -X_{i\sigma_{3}\otimes \sigma_{3}}]]$ 
and 
$-X_{i\sigma_3\otimes \sigma_3} [X_{I\otimes \xi_2}, 
  -X_{i\sigma_{3}\otimes \sigma_{3}}]]$ 
are vertical at $\Lambda$, 
(iv) 
$[X_{\xi_1\otimes I}, [X_{\xi_{1}\otimes I}, 
  -X_{i\sigma_{3}\otimes \sigma_{3}}]]$ 
and 
$[X_{I\otimes \xi_2}, [X_{I \otimes \xi_2}, 
  -X_{i\sigma_{3}\otimes \sigma_{3}}]]$ 
are vertical at $\Lambda$, 
and (v) 
$[X_{I \otimes \xi_2}, [X_{\xi_1\otimes I}, 
 -X_{i\sigma_{3}\otimes \sigma_{3}}]] = 
 [X_{\xi_1\otimes I}, [X_{I \otimes \xi_2}, 
  -X_{i\sigma_{3}\otimes \sigma_{3}}]] $.  
Since the derivative of the measure $F$ is given by
\begin{equation}
  dF 
  =
  8\mathrm{Re}\Bigl( 
    \det C^{\ast}
    \tr\left( \sigma_{2}C\sigma_{2}dC^{T} \right)
  \Bigr), 
  \label{dF}
\end{equation}
one has 
\begin{gather*}
  (dF)_{\Lambda}
  \Bigl(
    [X_{\xi_{1}\otimes I},
    [X_{I\otimes\xi_{2}}, -X_{i\sigma_{3}\otimes\sigma_{3}}]] 
  \Bigr) 
  = 
  \frac{-x_{1}y_{2}+y_{1}x_{2}}{2}
  \lambda_{1}\lambda_{2}(\lambda_{1}^{2}-\lambda_{2}^{2}). 
\end{gather*}
Hence, substituting this into \eqref{approximation of concurrence}, 
we have the following estimate of the concurrence, 
\begin{equation} 
  \sqrt{F(C(t))} 
  \sim 
  \sqrt{F(\varphi_{-i\sigma_{3}\otimes\sigma_{3}}^{Jt}(\Lambda))}
  + 
  \frac{Jt^{3}}{2\cdot 3!}
  (-x_{1}y_{2}+y_{1}x_{2})(\lambda_{1}^{2}-\lambda_{2}^{2}) 
  \label{approximation of F}
\end{equation}
for a sufficiently small $t>0$. 
It is to be noted that the effect of controls on 
the concurrence emerges actually at the third-order term in $t$. 
We have to note 
that for $\Lambda=\diag{1,0}$, 
one has $F(\varphi_{-i\sigma_{3}\otimes\sigma_{3}}^{Jt}(\Lambda))=0$, 
so that 
$\sqrt{F(C(t))} \sim \frac{Jt^{3}}{2\cdot 3!}(-x_{1}y_{2}+x_{2}y_{1})$. 
This means that the state $\Lambda = \diag{1,0}$ gets entangled slowly 
by the controls. 

The growth rate of the concurrence at the third order in $t$ depends on 
the quantity $-x_1y_2+y_1x_2$, 
where $x_{\alpha}, y_{\beta}$ are control parameters 
given by $\xi_{\alpha}=ix_{\alpha}\sigma_1+iy_{\alpha}\sigma_2$. 
It would be reasonable to restrict the magnitude of controls to 
$\sum_{\alpha=1}^2(x_{\alpha}^2+y_{\alpha}^2)=1$. 
If this is the case, the maximal growth ratio is realized when 
$x_1=-y_2$, $x_2=y_1$. 
This means that the most efficient control for entanglement is given by 
$\xi_1=ix_{1}\sigma_1+ix_{2}\sigma_2$, $\xi_2=ix_{2}\sigma_1-ix_{1}\sigma_2$ 
with $x_1^2+x_2^2=1/2$. 
However, if we choose the control given by 
$\xi_1=ix_{1}\sigma_1-ix_{2}\sigma_2$, $\xi_2=ix_{2}\sigma_1+ix_{1}\sigma_2$ 
with $x_1^2+x_2^2=1/2$, 
the two-qubit gets into a less entangled state. 
If $\xi_1=ix_{1}\sigma_1+ix_{1}\sigma_2$, 
$\xi_2=ix_{2}\sigma_1+ix_{2}\sigma_2$ 
with $x_1^2+x_2^2=1/2$, no change will occur in the entanglement of 
the two-qubit. 

We turn to the case that the initial state $C(0)$ is not diagonal.  
Let us be reminded of the fact that when the $C(0)$ is diagonal, 
the $\mathcal{O}(Jt^{2})$-term in \eqref{approximation of C} 
is vertical. 
However, if $C(0)$ is not diagonal, the $\mathcal{O}(Jt^{2})$-term is 
not vertical, so that Eqs.~\eqref{second order terms} and \eqref{dF} 
are put together to yield the contribution to entanglement promotion by 
\begin{align}
  & 
  \frac{Jt^{2}}{2\sqrt{F(C(0))}}
  (dF)_{C(0)}\biggl(
    [X_{\xi_{1}\otimes I}, -X_{i\sigma_{3}\otimes\sigma_{3}}]
    + 
   [X_{I\otimes\xi_{2}}, -X_{i\sigma_{3}\otimes\sigma_{3}}]
  \biggr) \notag\\
  = &  
  -\frac{Jt^{2}}{4}\mathrm{Re}\biggl( 
    e^{-i\theta} 
    (c_{01}+c_{10})
    (\overline{z_{1}}c_{00}+z_{2}c_{11}) 
  \biggr), 
  \label{second order terms 2}
\end{align}
where $z_{\alpha} = x_{\alpha} + iy_{\alpha}$ for $\alpha=1,2$, 
and where $e^{i\theta}$ is defined by 
$\det C = e^{i\theta}|\det C| = e^{i\theta}\sqrt{F(C)}$. 
From \eqref{second order terms 2}, 
the $\mathcal{O}(Jt^{2})$-term identically vanishes for arbitrary 
$\xi_{1}$ and $\xi_{2}$ 
if $c_{00}=c_{11}=0$ or $c_{01}+c_{10}=0$, 
but it takes a non-zero value in general, 
if $c_{01}+c_{10} \ne 0$ and $(c_{00},c_{11}) \ne (0,0)$.

\section{Concluding remarks}
We have discussed how the control and drift vector fields are coupled to 
make the NMR system more entangled on the two-qubit model, 
by examining the solution in detail in terms of horizontal and vertical 
vector fields and by evaluating the concurrence. 
As was pointed out in Sec. \!4.1,  the initial state 
$C(0) = \diag{e^{i\theta}, 0}$ or $e^{i\theta}|0\rangle\oplus |0\rangle$ 
never gets entangled without controls, but we have shown in Sec.~4.4 that 
it can be entangled if controls are taken. 
However, the effect of the control on entanglement is so slow that 
it emerges at the third-order term in $t$. 
In fact, as is shown in Sec.~4.4, only vertical vector fields are 
generated in the second-order term in $t$ by the coupling between 
the drift and control vector fields, and further vector fields with 
non-vanishing horizontal components emerge at the third-order term 
in $t$. 

It would be of help to state the relation of our 
theory to the Cartan decomposition of the Lie algebra. 
As for the two-qubit systems, the associated Lie algebra is 
$\mathfrak{su}(4)$, of which the Cartan decomposition 
is given by 
\begin{equation*}
  \mathfrak{su}(4)=\mathfrak{k}\oplus \mathfrak{p}, 
\end{equation*}
where 
\begin{align}
  \mathfrak{k} 
  = 
  \Span \{iI\otimes \sigma_j, i\sigma_k\otimes I \}_{j,k=1,2,3}, 
  \quad 
  \mathfrak{p}
  = 
  \Span \{2i\sigma_j \otimes \sigma_k \}_{j,k=1,2,3}. 
  \nonumber
\end{align}
The subalgebra $\mathfrak{k}$ generates vertical tangent vectors, 
and some of elements in $\mathfrak{p}$ span the horizontal subspace at 
$\Lambda=\diag{\lambda_1,\lambda_2}$. 
For example, if $\lambda_1\neq \lambda_2$, one has 
a horizontal vector $X_{i\sigma_1\otimes \sigma_2}(\Lambda)$ 
(see Props.~4.3 and 4.4). 
However, the vectors associated with $\mathfrak{p}$ are not always 
horizontal. 
As is seen from \eqref{sigma23 13}, 
$X_{i\sigma_2\otimes \sigma_3}(\Lambda)$ and 
$X_{i\sigma_1\otimes \sigma_3}(\Lambda)$ are 
vertical if $\Lambda \neq I/\sqrt{2}$. 
In \cite{ZVSW03}, as to non-local operation, 
they state that $\mathfrak{k}$ can be viewed as the 
local part in $\mathfrak{su}(4)$ and $\mathfrak{p}$ as the non-local part. 
In this sense, the horizontal vectors are of non-local nature, and 
the horizontal subspace makes the non-locality quite sharp. 
This is because the entanglement gets promoted most efficiently 
in the direction of horizontal vectors. 

We comment on the motivation behind this study, 
which is related to the decomposition method 
for the analysis of quantum control problems on compact Lie groups. 
The decomposition method is of much use to design 
a control in order to construct a desired unitary operator
\cite{KBG01, KG01, KGB02}, and 
then has been developed \cite{DAR06, DDAS08}. 
In these situations, the control (or local operation) works instantaneously, 
and the drift (or non-local) Hamiltonian is assumed to be negligible while 
the control works. 
Under this assumption, the time evolution operator 
$U_{t} = e^{-it\hat{H}}$ can be decomposed into 
the product of the local and non-local operators,  
\begin{equation}
  U_{t} = 
  L_{0} \,  e^{-i\tau_{1}\hat{H}_{d}} \, L_{1} 
  \cdots 
  e^{-i\tau_{N}\hat{H}_{d}} \, L_{N}
  \label{decomp of U_t}
\end{equation}
where 
$L_{0}, \cdots, L_{N} \in \SU{2^{\ell}}\otimes\SU{2^{m}}$ are generated by 
the control Hamiltonian, 
and where $\tau_{1}, \cdots, \tau_{N} > 0$ are time intervals for which 
the control Hamiltonian vanishes. 
In this view, the bipartite entanglement is not concerned by  
the coupling between the drift and control Hamiltonians in 
the time evolution. 
Our motivation is to consider the coupling in the evolution 
without the above assumption, and to investigate 
the relation of the coupling to the bipartite entanglement. 
The work \cite{Rom07} by Romano shares the motivation with us. 

In \cite{Rom07}, the decomposition of the unitary group $\SU{4}$ 
is used to study entanglement magnification. 
The time evolution $U_t$ is decomposed into $U_t=L_tA_tK_t$, where 
$L_t,K_t\in \SU{2}\otimes\SU{2}$ and $A_t\in H$ with 
$H$ the Cartan subgroup of $\SU{4}$. 
The author says that $L_t$ is irrelevant, and claims that 
the non-local operator $A_t$ makes a contribution 
to entanglement magnification together with a help of $K_{t}$. 
The decomposition method is sharp to study 
the entanglement in the evolution, 
because the evolution operator is completely separated into 
the ``entangling part'' $A_{t}$ and the two local operators $L_{t}$ and $K_{t}$. 
But, the process of entanglement promotion does not explicitly appear in 
the decomposition, since not only $A_{t}$ but also $L_{t}, K_{t}$ take in 
the coupling between the drift and control Hamiltonians. 
In fact, the $A_{t}$, $L_{t}$ and $K_{t}$ are all unknown until 
the time evolution operator $U_{t}$ is integrated in an explicit form. 

On the other hand, the method of this article is represented in 
a familiar form with bipartite entanglement, and 
explains how the coupling between the drift and control Hamiltonians occurs 
in the time evolution. 
See \eqref{N-param power series}, \eqref{expanded Tk}, 
\eqref{power series of J} and \eqref{approximation of C}. 
In particular, the right hand-side of \eqref{approximation of C} is not 
discussed in the context of \eqref{decomp of U_t}. 
In a comparison to \cite{Rom07}, 
our method allows one not to solve the time evolution of the system, 
although it is not sharp in the sense that the terms in \eqref{approximation of C} contain 
vertical (or local) components. 

The formulas, \eqref{N-param power series} together with 
\eqref{expanded Tk}, and \eqref{approximation of C} also show 
how the effect of the controls spreads among qubits in 
the dynamical system. 
In contrast with this, the proof of the controllability of the NMR system 
would show the spread of the controls among qubits in a ``static'' situation. 
For the sake of self-containedness, we give the proof of 
the controllability of our NMR system (see Theorem \ref{Thm: controllability}) 
in a method different from that in \cite{AD'A02}. 
As is pointed out in App.~B, 
the NMR system is controllable if and only if the associated spin graph is 
connected.

\section*{Acknowledgments}
The authors would like to thank to the referees for valuable comments and 
suggestions which led to much improvement of this article. 
The second author is financially supported by 
a JSPS Research Fellowship for Young Scientists, 19-3956.

\appendix
\section{Lie brackets in the two-qubit NMR system}
A straightforward calculation along with Corollary 
\ref{Prop: comp of brackets 2} provides 
\begin{align*} 
  &
  [X_{i\sigma_{3}\otimes\sigma_{3}}, X_{i\sigma_{1}\otimes I}] 
  = 
  -X_{i\sigma_{2}\otimes\sigma_{3}}, 
  & 
  [X_{i\sigma_{3}\otimes\sigma_{3}}, X_{i\sigma_{2}\otimes I}] 
  = 
  X_{i\sigma_{1}\otimes\sigma_{3}}, \\
  &
  [X_{i\sigma_{3}\otimes\sigma_{3}}, X_{iI\otimes\sigma_{1}}] 
  = 
  -X_{i\sigma_{3}\otimes\sigma_{2}}, 
  & 
  [X_{i\sigma_{3}\otimes\sigma_{3}}, X_{iI\otimes\sigma_{2}}] 
  = 
  X_{i\sigma_{3}\otimes\sigma_{1}}, \\
  & 
  [X_{i\sigma_{1}\otimes I}, X_{i\sigma_{2}\otimes I}] 
  = 
  -X_{i\sigma_{3}\otimes I}, 
  & 
  [X_{iI\otimes\sigma_{1}}, X_{iI\otimes\sigma_{2}}] 
  = 
  -X_{iI\otimes\sigma_{3}}, 
\end{align*}
and the other Lie brackets prove to vanish. 
Further, double Lie brackets are given as follows: 
\begin{align*}
  & 
  [X_{i\sigma_{3}\otimes\sigma_{3}}, 
  [X_{i\sigma_{3}\otimes\sigma_{3}},  X_{i\sigma_{1}\otimes I}]]
  = 
  -[X_{i\sigma_{3}\otimes\sigma_{3}}, X_{i\sigma_{2}\otimes\sigma_{3}}] 
  = 
  -X_{i\sigma_{1}\otimes I}, \\
  & 
  [X_{i\sigma_{3}\otimes\sigma_{3}}, 
  [X_{i\sigma_{3}\otimes\sigma_{3}},  X_{i\sigma_{2}\otimes I}]]
  = \quad \!
  [X_{i\sigma_{3}\otimes\sigma_{3}}, X_{i\sigma_{1}\otimes\sigma_{3}}]
  = 
  -X_{i\sigma_{2}\otimes I}, \\
  & 
  [X_{i\sigma_{3}\otimes\sigma_{3}}, 
  [X_{i\sigma_{3}\otimes\sigma_{3}},  X_{iI\otimes\sigma_{1}}]] 
  = 
  -[X_{i\sigma_{3}\otimes\sigma_{3}}, X_{i\sigma_{3}\otimes\sigma_{2}}]
  = 
  -X_{iI\otimes\sigma_{1}}, \\
  & 
  [X_{i\sigma_{3}\otimes\sigma_{3}}, 
  [X_{i\sigma_{3}\otimes\sigma_{3}},  X_{iI\otimes\sigma_{2}}]] 
  = \quad \!
  [X_{i\sigma_{3}\otimes\sigma_{3}}, X_{i\sigma_{3}\otimes\sigma_{1}}]
  = 
  -X_{iI\otimes\sigma_{2}}, \\
  & 
  [X_{i\sigma_{1}\otimes I}, \,\,
  [X_{i\sigma_{3}\otimes\sigma_{3}}, \, X_{i\sigma_{1}\otimes I}]] 
  = 
  - [X_{i\sigma_{1}\otimes I}, \,\, X_{i\sigma_{2}\otimes\sigma_{3}}] 
  = \quad \!
  X_{i\sigma_{3}\otimes\sigma_{3}}, \\
  & 
  [X_{i\sigma_{1}\otimes I}, \,\,
  [X_{i\sigma_{3}\otimes\sigma_{3}}, \, X_{i\sigma_{2}\otimes I}]] 
  = \quad \!
  [X_{i\sigma_{1}\otimes I}, \,\, X_{i\sigma_{1}\otimes\sigma_{3}}] 
  = \quad \! 
  0, \\
  & 
  [X_{i\sigma_{1}\otimes I}, \,\, 
  [X_{i\sigma_{3}\otimes\sigma_{3}}, \, X_{iI\otimes\sigma_{1}}]] 
  = 
  -[X_{i\sigma_{1}\otimes I}, \,\, X_{i\sigma_{3}\otimes\sigma_{2}}]
  = 
  -X_{i\sigma_{2}\otimes\sigma_{2}}, \\
  & 
  [X_{i\sigma_{1}\otimes I}, \,\,
  [X_{i\sigma_{3}\otimes\sigma_{3}}, \, X_{iI\otimes\sigma_{2}}]] 
  = \quad \!
  [X_{i\sigma_{1}\otimes I}, \,\, X_{i\sigma_{3}\otimes\sigma_{1}}]
  = \quad \! 
  X_{i\sigma_{2}\otimes\sigma_{1}}, \\
  & 
  [X_{i\sigma_{2}\otimes I}, \,\,
  [X_{i\sigma_{3}\otimes\sigma_{3}}, \, X_{i\sigma_{1}\otimes I}]] 
  = 
  -[X_{i\sigma_{2}\otimes I}, \,\, X_{i\sigma_{2}\otimes\sigma_{3}}] 
  = \quad \!
  0, \\
  & 
  [X_{i\sigma_{2}\otimes I}, \,\,
  [X_{i\sigma_{3}\otimes\sigma_{3}}, \, X_{i\sigma_{2}\otimes I}]] 
  = \quad \!
  [X_{i\sigma_{2}\otimes I}, \,\, X_{i\sigma_{1}\otimes\sigma_{3}}] 
  = \quad \! 
  X_{i\sigma_{3}\otimes\sigma_{3}}, 
 \\
  & 
  [X_{i\sigma_{2}\otimes I}, \,\, 
  [X_{i\sigma_{3}\otimes\sigma_{3}}, \, X_{iI\otimes\sigma_{1}}]] 
  = 
  -[X_{i\sigma_{2}\otimes I}, \,\, X_{i\sigma_{3}\otimes\sigma_{2}}]
  = \quad \!
  X_{i\sigma_{1}\otimes\sigma_{2}}, \\
  & 
  [X_{i\sigma_{2}\otimes I}, \,\,
  [X_{i\sigma_{3}\otimes\sigma_{3}}, \, X_{iI\otimes\sigma_{2}}]] 
  = \quad \!
  [X_{i\sigma_{2}\otimes I}, \,\, X_{i\sigma_{3}\otimes\sigma_{1}}]
  = 
  -X_{i\sigma_{1}\otimes\sigma_{1}}, 
  \\
  & 
  [X_{iI\otimes\sigma_{1}}, \,\,
  [X_{i\sigma_{3}\otimes\sigma_{3}}, \, X_{i\sigma_{1}\otimes I}]] 
  = 
  -[X_{iI\otimes\sigma_{1}}, \,\, X_{i\sigma_{2}\otimes\sigma_{3}}] 
  = 
  -X_{i\sigma_{2}\otimes\sigma_{2}}, \\
  & 
  [X_{iI\otimes\sigma_{1}}, \,\,
  [X_{i\sigma_{3}\otimes\sigma_{3}}, \, X_{i\sigma_{2}\otimes I}]] 
  = \quad \!
  [X_{iI\otimes\sigma_{1}}, \,\, X_{i\sigma_{1}\otimes\sigma_{3}}] 
  = \quad \!  
  X_{i\sigma_{1}\otimes\sigma_{2}}, \\
  & 
  [X_{iI\otimes\sigma_{1}}, \,\, 
  [X_{i\sigma_{3}\otimes\sigma_{3}}, \, X_{iI\otimes\sigma_{1}}]] 
  = 
  -[X_{iI\otimes\sigma_{1}}, \,\, X_{i\sigma_{3}\otimes\sigma_{2}}]
  = \quad \! 
  X_{i\sigma_{3}\otimes\sigma_{3}}, 
  \\
  & 
  [X_{iI\otimes\sigma_{1}}, \,\,
  [X_{i\sigma_{3}\otimes\sigma_{3}}, \, X_{iI\otimes\sigma_{2}}]] 
  = \quad \!
  [X_{iI\otimes\sigma_{1}}, \,\, X_{i\sigma_{3}\otimes\sigma_{1}}]
  = \quad \! 
  0 \\
  & 
  [X_{iI\otimes\sigma_{2}}, \,\,
  [X_{i\sigma_{3}\otimes\sigma_{3}}, \, X_{i\sigma_{1}\otimes I}]] 
  = 
  -[X_{iI\otimes\sigma_{2}}, \,\, X_{i\sigma_{2}\otimes\sigma_{3}}] 
  = \quad \! 
  X_{i\sigma_{2}\otimes\sigma_{1}}, \\
  & 
  [X_{iI\otimes\sigma_{2}}, \,\,
  [X_{i\sigma_{3}\otimes\sigma_{3}}, \, X_{i\sigma_{2}\otimes I}]] 
  = \quad \!
  [X_{iI\otimes\sigma_{2}}, \,\, X_{i\sigma_{1}\otimes\sigma_{3}}] 
  = 
  -X_{i\sigma_{1}\otimes\sigma_{1}}, \\
  & 
  [X_{iI\otimes\sigma_{2}}, \,\, 
  [X_{i\sigma_{3}\otimes\sigma_{3}}, \, X_{iI\otimes\sigma_{1}}]] 
  = 
  -[X_{iI\otimes\sigma_{2}}, \,\, X_{i\sigma_{3}\otimes\sigma_{2}}]
  =\quad \! 
  0, \\
  & 
  [X_{iI\otimes\sigma_{2}}, \,\,
  [X_{i\sigma_{3}\otimes\sigma_{3}}, \, X_{iI\otimes\sigma_{2}}]] 
  = \quad \!
  [X_{iI\otimes\sigma_{2}}, \,\, X_{i\sigma_{3}\otimes\sigma_{1}}]
  = \quad \! 
  X_{i\sigma_{3}\otimes\sigma_{3}}. 
\end{align*}

\section{Controllability of the NMR system}
As stated in Sec. \!\!\!3, we study the controllability of our 
NMR system. 
The drift Hamiltonian we treat is given by 
\begin{equation}
  \label{free-H}
  \hat{H}_d 
  = 
  \hat{H}_0
  =\sum_{1\leq\alpha<\beta\leq n} \, 
  J_{\alpha\beta}\, 
  \sigma^{(\alpha)}_3\sigma^{(\beta)}_3, 
\end{equation}
where, for simplicity, 
we have denoted the drift Hamiltonian by $\hat{H}_0$ and 
abbreviated $\sigma_{j}^{(\alpha,n)}$ to $\sigma_{j}^{(\alpha)}$, 
where $\sigma_{j}^{(\alpha,n)}$ were defined in \eqref{sigma}.  
For our convenience, we express the control Hamiltonian as follows: 
\begin{equation*}
  \hat{H}_c
  =
  \sum_{\mu=1}^{n} v_1^{(\mu)}\hat{H}_{\mu} 
  + 
  \sum_{\mu=n+1}^{2n}v_2^{(\mu-n)}\hat{H}_{\mu}, 
\end{equation*}
where 
\begin{equation}
  \label{control-H}
   \hat{H}_{\mu}
  := 
  \begin{cases} 
    \sigma_1^{(\mu)}, & \mu = 1,\cdots, n, \\
    \sigma_2^{(\mu-n)}, & \mu = n+1,\cdots, 2n, 
  \end{cases}
\end{equation}
Our control system with the Hamiltonian 
$\hat{H}=\hat{H}_d+\hat{H}_c$ is (pure state) controllable 
if and only if the Lie algebra generated by 
$i\hat{H}_0,i\hat{H}_{\mu},\,\mu=1,\cdots,2n$, 
is equal to $\mathfrak{su}(2^n)$. 

For comparison's sake, we here quote the drift and control Hamiltonians 
studied in [1], which are given by 
\begin{equation}
  \label{d-free-H}
  \hat{H}'_0 
  = 
  \sum_{\alpha<\beta}(
    M_{\alpha\beta}\sigma_1^{(\alpha)}\sigma_1^{(\beta)} 
    + 
    N_{\alpha\beta}\sigma_2^{(\alpha)}\sigma_2^{(\beta)} 
    + 
    P_{\alpha\beta}\sigma_3^{(\alpha)}\sigma_3^{(\beta)}
  ),
\end{equation}
and 
\begin{equation}
  \label{d-control-H}
  \hat{H}'_{j} 
  = 
  \sum_{\alpha=1}^{n} r_{\alpha} \sigma_{j}^{(\alpha)}, \quad 
  j =1,2,3, 
\end{equation} 
respectively, where 
$M_{\alpha\beta}, N_{\alpha\beta}, P_{\alpha\beta} \in \mathbb{R}$ 
are coupling constants 
and $r_{\alpha} \in \mathbb{R}$ are the gyromagnetic ratio 
of the $\alpha$-th spin-$\frac{1}{2}$ particle (or qubit). 
In our case of \eqref{control-H}, it is assumed that 
each of spin-$\frac12$ particles can be stimulated 
by the only two components of the magnetic field. 
In contrast with this, in the case of (\ref{d-control-H}), 
all the spin-$\frac12$ particles are assumed to be 
stimulated simultaneously by all the three components of the 
magnetic field with possibly different gyromagnetic ratios. 
While our drift Hamiltonian \eqref{free-H} contains only 
$\sigma_3$ factors, but theirs \eqref{d-free-H} has all the $\sigma_j$ factors. 
In \cite{AD'A02}, they prove that the NMR system with the Hamiltonian 
$\hat{H}'_0+\sum_{j=1}^{3} u_{j}\hat{H}'_{j}$ is controllable if and only if 
the associated spin graph is connected. 

Our objective in what follows is to show that our NMR system is controllable 
if and only if the associated spin graph is connected. 
The main point is to show that the Lie algebra generated by 
the operators $-i\hat{H}_\mu,\, \mu = 0,1,\cdots,2n$, 
is equal to $\mathfrak{su}(2^n)$. 
In this respect, what we have to do for proof is the same as that 
in \cite{AD'A02}. 
However, as our Hamiltonian is different from that in \cite{AD'A02}, 
the method for proof should be different from that in \cite{AD'A02}.

\subsection{A decomposition of $\frak{su}(2^{n})$} 
Before calculating commutators among 
$\{-i\hat{H}_\mu\}_{\mu=0,1,\cdots,2n}$, 
we decompose $\frak{su}(2^{n})$ into the sum of subspaces. 
We denote by $\mathcal{B}$ the Lie algebra $\frak{su}(2^{n})$, 
\begin{equation}
  \mathcal{B}
  =
  \Span_{\mathbb{R}}\{
    i\sigma_{j_{1}}\otimes \sigma_{j_{2}} \otimes \cdots \otimes 
    \sigma_{j_{n}}  \,|\,
  j_{1}, j_{2}, \cdots, j_{n}\in \{0,1,2,3\} 
  \} 
  \setminus \{iI^{\otimes n}\},  
  \label{thebasis}
\end{equation}
and break it up into subspaces $\mathcal{B}^{(k)},\,k=0,1,\cdots,n-1$, 
all the basis elements of which have as many as $k$ $I$-factors. 
Since $\dim\mathcal{B}^{(k)}={n \choose k}3^{n-k}$, and 
since $\dim\mathfrak{su}(4)=4^n-1 =\sum_{k=0}^{n-1}{n \choose k}3^{n-k}$, 
the vector space $\mathcal{B}$ is broken up into 
\begin{equation}
  \label{decomp-B}
  \mathcal{B} = \bigoplus_{k=0}^{n-1}\mathcal{B}^{(k)}.
\end{equation}
For $k=0$, $\mathcal{B}^{(0)}$ is expressed as 
\begin{equation*}
  \mathcal{B}^{(0)}
  =
  \Span_{\mathbb{R}}\{
    i\sigma_{j_{1}}\otimes \sigma_{j_{2}} \otimes \cdots 
    \otimes \sigma_{j_{n}}
  \,|\, 
    j_{1}, j_{2}, \cdots, j_{n} \in \{1, 2, 3\}
  \}. 
\end{equation*}
Among $\mathcal{B}^{(k)}$ with $k=1,\cdots, n,$ 
we write down $\mathcal{B}^{(k)}$ only for $k=n-1$; 
\begin{equation*}
  \mathcal{B}^{(n-1)}
  =
  \Span_{\mathbb{R}}\{
    i\sigma_{j}^{(1)}, 
    i\sigma_{j}^{(2)},
    \cdots,
    i\sigma_{j}^{(n)}
  \,|\;
    j = 1, 2, 3 
  \}. 
\end{equation*}
Further, we set 
\begin{equation}
  \label{Bz}
  \mathcal{B}_3
  :=
  \Span_{\mathbb{R}}\{
  i\sigma_{j_{1}}\otimes \sigma_{j_{2}} \otimes \cdots \otimes \sigma_{j_{n}} 
  \,|\,
  j_{1}, j_{2}, \cdots, j_{n} \in \{0, 3 \} 
  \}
  \setminus
  \{iI^{\otimes n}\}. 
\end{equation}
The following lemma shows that $\mathcal{B}_{3}$ holds a key position in 
calculating commutators. 

\begin{Lem}
\label{Lemma i}
Commutators among 
$\{-i\hat{H}_\mu\}_{\mu=1,\cdots,2n}\cup \mathcal{B}_3$ 
generate all the basis operators of $\mathcal{B}$. 
\end{Lem}

\begin{proof}
Note that we have 
$\mathcal{B}_3=\sum_{k=0}^{n-1}(\mathcal{B}_3\cap \mathcal{B}^{(k)})$ 
from \eqref{decomp-B} and \eqref{Bz}. 
For $k=0$, we take an operator 
$i\sigma_3^{\otimes n}\in \mathcal{B}_3 \cap \mathcal{B}^{(0)}$. 
The commutators between $i\sigma_3^{\otimes n}$ and $-i\hat{H}_\mu$ are 
given by 
\begin{equation*}
  [i\sigma_3^{\otimes n},-i\hat{H}_\mu]
  =
  \begin{cases}
    i\sigma_3^{\otimes(\mu-1)} \otimes \sigma_2 \otimes 
    \sigma_3^{\otimes(n-\mu)}, 
    & 
    \mu = 1,\cdots,n, \\ 
    -i\sigma_3^{\otimes(\mu-n-1)} \otimes \sigma_1 \otimes 
    \sigma_3^{\otimes(2n-\mu)}, 
    & 
    \mu = n+1,\cdots,2n. 
  \end{cases}
\end{equation*}
Taking successive commutators 
$[[i\sigma_3^{\otimes n},-i\hat{H}_\mu],-i\hat{H}_\nu]$, 
$\mu, \nu = 1,\cdots,2n$, and so on, 
we can obtain the operators of the form, up to $\pm$ sign,  
\begin{equation*}
  i\sigma_{j_{1}}\otimes \sigma_{j_{2}} \otimes \cdots \otimes \sigma_{j_{n}}, 
  \quad 
  j_{1}, j_{2}, \cdots, j_{n} \in \{1, 2, 3 \}. 
\end{equation*}
If we start with 
$iI\otimes \sigma_3^{\otimes(n-1)}\in \mathcal{B}_3\cap \mathcal{B}^{(1)}$ and 
follow the same procedure as above, we can obtain 
\begin{equation*}
  iI\otimes \sigma_{j_{2}} \otimes \cdots \otimes \sigma_{j_{n}}, 
  \quad 
  j_{2}, \cdots, j_{n} \in \{1, 2, 3 \}. 
\end{equation*}
Through the process of taking the commutator of an operator in 
$\mathcal{B}_3\cap \mathcal{B}^{(k)}$ with $-i\hat{H}_\mu$, 
the resultant tensor product operator may have 
the $I$-factors fixed and $\sigma_3$-factors in the 
site $\mu$ or $\mu-n$ changed to $\sigma_1$ or $\sigma_2$, according to 
whether $\mu\in\{1,\cdots,n\}$, or $\mu\in \{n+1,\cdots,2n\}$.  
Taking successive commutators, 
we can obtain the tensor product operators 
with the $(n-k)$ $\sigma_3$-factors replaced by 
$\sigma_1$ or $\sigma_2$ and the $k$ $I$-factors fixed, so that 
we have all the basis elements of $\mathcal{B}^{(k)}$ 
in the form of commutators among 
$\{-i\hat{H}_\mu\}_{\mu=1,\cdots,2n}
\cup(\mathcal{B}_3\cap \mathcal{B}^{(k)})$, 
$k=0,1,\cdots,n-1$. 
This ends the proof. 
\end{proof}

\subsection{Controllability with the complete spin graph}
We now assume that $J_{\alpha\beta}\neq 0$ 
for all $1\leq \alpha < \beta \leq n$, 
or that the associated spin graph is complete, 
to show that the commutators among $\{-i\hat{H}_\mu\}_{\mu=0,1\cdots,2n}$ 
span $\mathfrak{su}(2^n)$. 
On account of Lemma \ref{Lemma i}, 
we have only to show that 
commutators among $\{-i\hat{H}_\mu\}_{\mu=0,1,\cdots,2n}$ 
can span $\mathcal{B}_3$. 
We first show that by taking commutators among 
$\{-i\hat{H}_\mu\}_{\mu=0,1,\cdots,2n}$, 
the term $-i\sigma^{(\alpha)}_3\sigma^{(\beta)}_3$ of the drift Hamiltonian 
$-i\hat{H}_0$ can be singled out, if $J_{\alpha\beta}\neq 0$. 
For 
$-i\hat{H}_0=
-i\sum_{\alpha<\beta}J_{\alpha\beta}\sigma^{(\alpha)}_3\sigma^{(\beta)}_3$ 
and 
$-i\hat{H}_\mu=-i\sigma_1^{(\mu)}
,\,\mu=1,\cdots,n-1$, 
the commutators between them 
are given by 
\begin{equation*}
  [-i\hat{H}_0,-i\hat{H}_\mu] 
  = 
  -i\sum_{\alpha=1}^{\mu-1}
  J_{\alpha\mu}\sigma^{(\alpha)}_3\sigma^{(\mu)}_2 
  -i\sum_{\beta=\mu+1}^{n}
  J_{\mu\beta}\sigma^{(\mu)}_2\sigma^{(\beta)}_3. 
\end{equation*}
We calculate further commutators to obtain, 
for $1\leq \mu<\nu\leq n$,  
\begin{align*}
  [[-i\hat{H}_0, -i\hat{H}_{\mu}], -i\hat{H}_{\nu}] 
  & =
  -iJ_{\mu\nu}\sigma_2^{(\mu)}\sigma_2^{(\nu)}. 
\end{align*}
Moreover, a calculation with the right-hand side of the above provides, 
for $1\leq \mu < \nu\leq n$, 
\begin{equation*}
  [-i\hat{H}_\mu, [-i\hat{H}_{\nu},
  -iJ_{\mu\nu}\sigma_2^{(\mu)}\sigma_2^{(\nu)}]] 
  = 
  [-i\sigma_1^{(\mu)},-iJ_{\mu\nu}\sigma_2^{(\mu)}\sigma_3^{(\mu)}] 
  = 
  -iJ_{\mu\nu}\sigma_3^{(\mu)}\sigma_3^{(\nu)}. 
\end{equation*}
If $J_{\mu\nu}\neq 0$, we obtain $-i\sigma_3^{(\mu)}\sigma_3^{(\nu)}$ 
as a commutator among 
$\{i\hat{H}_{0}, \cdots, i\hat{H}_{2n}\}$. 
Thus, we have proved the following

\begin{Lem}
\label{Lemma ii}
If $J_{\alpha\beta}\neq 0$ for $\alpha$ and $\beta$ with 
$1\leq \alpha< \beta \leq n$, 
the operator $-i\sigma_3^{(\alpha)}\sigma_3^{(\beta)}$ 
can be realized as a commutator among 
$\{-i\hat{H}_\mu\}_{\mu=0,1,\cdots,2n}$. 
\end{Lem}

We now show that the $\mathcal{B}_3$ can be generated 
from $\{-i\hat{H}_\mu\}_{\mu=0,1,\cdots,2n}$, 
by using Lemma \ref{Lemma ii} under 
the assumption that $J_{\alpha\beta}\neq 0$ for 
all $1\leq \alpha<\beta\leq n$. 
Let $\alpha_i,\,i=1,\cdots,m$, be positive integers such that 
\begin{equation}
  \label{alpha-m}
  1 \leq \alpha_1 < \alpha_2 
  < \cdots < 
  \alpha_{m-1} < \alpha_{m} \leq n 
  \quad 
  \mbox{with}\; m \geq 2.
\end{equation}
For a given sequence of $\alpha_i$, 
Lemma \ref{Lemma ii} provides us with the operators 
\begin{equation*}
  -i\sigma_3^{(\alpha_1)} \sigma_3^{(\alpha_2)},\;
  -i\sigma_3^{(\alpha_2)} \sigma_3^{(\alpha_3)},\;
  \cdots,  
  -i\sigma_3^{(\alpha_{m-1})} \sigma_3^{(\alpha_m)}
\end{equation*}
in the form of commutators among 
$\{-i\hat{H}_\mu\}_{\mu=0,1,\cdots,2n}$. 
We can form commutators among these operators and 
$\{-i\hat{H}_\mu\}_{\mu=0,1,\cdots,2n}$ to obtain 
\begin{align*}
  F_1 
  &:= 
  -[
    -i\sigma_3^{(\alpha_1)}\sigma_3^{(\alpha_2)}, 
    -i\hat{H}_{n+\alpha_2}
  ] 
  =
  -i\sigma_3^{(\alpha_1)}\sigma_1^{(\alpha_2)}, \\
  F_{k} 
  &:= 
  -[
    [
      -i\sigma_{3}^{(\alpha_{k})} \sigma_3^{(\alpha_{k+1})}, 
      -i\hat{H}_{\alpha_{k}}
    ], 
    -i\hat{H}_{n+\alpha_{k+1}}
  ] 
  = 
  -i\sigma_2^{(\alpha_{k})}\sigma_1^{(\alpha_{k+1})}, \\
  F_{m-1} 
  & := 
  [
    -i\sigma_3^{(\alpha_{m-1})}\sigma_3^{(\alpha_{m})},
    -i\hat{H}_{\alpha_{m-1}}
  ] 
  = 
  -i\sigma_2^{(\alpha_{m-1})}\sigma_3^{(\alpha_{m})}, 
\end{align*}
where $k=2,\cdots, m-2$, if $m\geq 3$. 
From $F_1,F_2,\cdots,F_{m-1}$, 
we can form commutators, for $m=3,\cdots,n$, 
\begin{align*}
  C_m 
  := 
  [\cdots [[F_1,F_2], \cdots, F_{m-2}],F_{m-1} ]
  =
  -i\sigma^{(\alpha_1)}_3 \sigma^{(\alpha_2)}_3 
  \cdots 
  \sigma^{(\alpha_{m-1})}_3 \sigma^{(\alpha_m)}_3 ,
\end{align*}
which is in $\mathcal{B}_3 \cap \mathcal{B}^{(n-m)}$. 
For $m=1,2$, the operators $C_1,C_2$ can be expressed, respectively, as 
\begin{align}
  C_1 
  &:= 
  -i\sigma_3^{(\alpha)}
  \;=\; 
  [-i\hat{H}_{\alpha},-i\hat{H}_{n+\alpha}], \quad 
  1\leq \alpha \leq n, \\
  C_2 
  & := 
  -i\sigma_3^{(\alpha_1)}\sigma_3^{(\alpha_2)},
\end{align}
where $C_1\in \mathcal{B}_3 \cap \mathcal{B}^{(n-1)}$ and  
$C_2\in \mathcal{B}_3 \cap \mathcal{B}^{(n-2)}$, and these can be expressed as 
commutators among $\{-i\hat{H}_\mu\}_{\mu=0,1,\cdots,2n}$ as well.

Taking all possible sequences of $\alpha_i$ subject to (\ref{alpha-m}), 
we can construct all possible tensor product operators 
with $m$ $\sigma_3$-factors and $(n-m)$ $I$-factors 
as commutators among $\{-i\hat{H}_\mu\}_{\mu=0,1,\cdots,2n}$. 
In other words, we can form all the basis elements of 
$\mathcal{B}_3 \cap \mathcal{B}^{(n-m)}$ 
in the form of commutators among 
$\{-i\hat{H}_\mu\}_{\mu=0,1,\cdots,2n}$. 
Since 
$\mathcal{B}_3
=
\sum_{m=1}^n (\mathcal{B}_3 \cap \mathcal{B}^{(n-m)})$,   
we obtain the following

\begin{Lem}
\label{Lemma iii} 
If $J_{\alpha\beta}\neq 0$ for all $\alpha$ and $\beta$ 
with $1\leq \alpha <\beta \leq n$, 
the $\mathcal{B}_3$ are generated by commutators among 
$\{-i\hat{H}_\mu\}_{\mu=0,1,\cdots,2n}$. 
\end{Lem}

From Lemmas \ref{Lemma i} and \ref{Lemma iii}, 
it turns out that if $J_{\alpha\beta}\neq 0$ 
for all $\alpha$ and $\beta$ with 
$1\leq \alpha < \beta \leq n$, 
$\mathcal{B}$ is generated by taking commutators among 
$\{-i\hat{H}_\mu\}_{\mu=0,1,\cdots,2n}$. 
Hence, we obtain 

\begin{Prop}
\label{Proposition iv}
If $J_{\alpha\beta}\neq 0$ for all $\alpha$ and $\beta$ 
with $1\leq \alpha < \beta\leq n$, 
or if the spin graph associated with the drift Hamiltonian is complete, 
the Lie algebra $\mathfrak{su}(2^n)$ is generated from 
$\{-i\hat{H}_\mu\}_{\mu=0,1,\cdots,2n}$. 
\end{Prop}

\subsection{Controllability with a connected spin graph}
In the following, we consider the case where some of 
$J_{\alpha\beta}$'s may vanish, but the spin graph is connected. 
Suppose that 
$J_{\alpha\beta},J_{\beta\gamma}\neq 0$, and $J_{\alpha\gamma}=0$.  
We may assume that $\alpha<\beta<\gamma$. 
By Lemma \ref{Lemma ii}, we can put the operators 
$-i\sigma_3^{(\alpha)}\sigma_3^{(\beta)},\, 
-i\sigma_3^{(\beta)}\sigma_3^{(\gamma)}$ 
in the form of commutators among 
$\{-i\hat{H}_\mu\}_{\mu=0,1,\cdots,2n}$. 
We then take commutators to obtain 
\begin{align}
  D 
  & := 
  [[ -i\sigma_3^{(\alpha)}\sigma_3^{(\beta)},-i\hat{H}_{\beta}],
  -i\sigma_3^{(\beta)}\sigma_3^{(\gamma)}] 
  = 
  -i\sigma_3^{(\alpha)}\sigma_1^{(\beta)}\sigma_3^{(\gamma)}. 
\end{align}
From $D, -i\sigma_3^{(\beta)}\sigma_3^{(\gamma)}, -i\hat{H}_{n+\beta}$, 
and $-i\hat{H}_\gamma$, we obtain 
\begin{align}
  D' 
  & := 
  4[[D,-i\hat{H}_\gamma],
  [-i\hat{H}_{n+\beta},-i\sigma_3^{(\beta)}\sigma_3^{(\gamma)}]] 
  = 
  -i\sigma_3^{(\alpha)}\sigma_1^{(\gamma)}.
\end{align}
Further calculation provides 
\begin{equation*}
  D'':= [D',-i\hat{H}_{n+\gamma}] 
  = 
  -i\sigma_3^{(\alpha)}\sigma_3^{(\gamma)}. 
\end{equation*}
Thus we have found that the operators 
$-i\sigma_3^{(\alpha)}\sigma_3^{(\gamma)}$ can be described as  
commutators taken among $-i\sigma_3^{(\alpha)}\sigma_3^{(\beta)}$, 
$-i\sigma_3^{(\beta)}\sigma_3^{(\gamma)}$, 
and $-i\hat{H}_\mu,\,\mu=1,\cdots,2n$. 
This implies that if there are interactions between particles 
$\alpha$ and $\beta$ and between particles $\beta$ and $\gamma$, 
an interaction between particles 
$\alpha$ and $\gamma$ is induced, 
though $J_{\alpha\gamma}=0$ at the beginning, 
by taking commutators among $\{-i\hat{H}_\mu\}_{\mu=0,1,\cdots,2n}$. 
We may interpret this fact as follows: 
In the process of making commutators among 
$\{-i\hat{H}_\mu\}_{\mu=0,1,\cdots,2n}$, 
two edges $(\alpha,\beta), (\beta,\gamma)$ of 
the spin graph give rise to 
a new edge $(\alpha,\gamma)$ which represents an interaction between 
two particles $\alpha$ and $\gamma$.

So far we have studied the case where there are two edges between  
the nodes $\alpha$ and $\gamma$. 
We may apply the same procedure to the case where there are 
$r$ edges between the nodes $\alpha$ and $\gamma$, 
where $r=2,\cdots,n-1$. 
Suppose we have, say, three edges 
$(\alpha,\beta_1),(\beta_1,\beta_2), (\beta_2,\gamma)$ 
with 
$J_{\alpha\beta_1},J_{\beta_1\beta_2},J_{\beta_2 \gamma}\neq 0$. 
The same procedure as above yields the edge $(\alpha,\beta_2)$ 
from $(\alpha,\beta_1)$ and $(\beta_1,\beta_2)$, and consequently 
$(\alpha,\gamma)$ from $(\alpha,\beta_2)$ and $(\beta_2,\gamma)$. 
Then, it turns out that 
if two particles $\alpha$ and $\gamma$ are linked with a sequence of 
two-particle interactions, the operator 
$-i\sigma_3^{(\alpha)}\sigma_3^{(\gamma)}$ 
describing an interaction between the particles $\alpha$ and $\gamma$ 
is induced by taking commutators among 
$\{-i\hat{H}_\mu\}_{\mu=0,1,\cdots,2n}$. 
Since the spin graph is connected, there is a sequence of 
edges between any pair of nodes $\alpha$ and $\gamma$, so that 
one can obtain the operator 
$-i\sigma_3^{(\alpha)}\sigma_3^{(\gamma)}$ expressed as 
a commutator among $\{-i\hat{H}_\mu\}_{\mu=0,1,\cdots,2n}$. 
Thus we have reached the same conclusion 
as in Lemma \ref{Lemma ii} with an 
assumption weaker than that in Lemma \ref{Lemma ii}. 

\begin{Lem}
\label{Lemma v} 
If the associated spin graph is connected, 
all the interaction operators 
$-i\sigma_3^{(\alpha)}\sigma_3^{(\beta)}$, 
$1\leq \alpha<\beta\leq n$, are generated from 
the drift and control Hamiltonians. 
\end{Lem}

Thus Prop. \ref{Proposition iv} is refined as follows: 

\begin{Prop}
\label{Proposition vi} 
If the spin graph associated with the drift 
Hamiltonian is connected, 
the Lie algebra $\mathfrak{su}(2^n)$ is generated from 
$\{-i\hat{H}_\mu\}_{\mu=0,1,\cdots,2n}$. Thus, the NMR system is 
controllable if the associated spin graph is connected. 
\end{Prop}

\subsection{A necessary and sufficient condition for controllability}
We now consider the case where the spin graph $S$ is disconnected. 
We assume that the graph $S$ is broken up into two disjoint subgraphs, 
$S_1$ and $S_2$, each of which is connected. 
Suppose that $S_1$ and $S_2$ have nodes $1$ to $r$, and 
$r+1$ to $n$, respectively. 
According to the decomposition, $S=S_1\cup S_2$, of the spin graph, 
the drift Hamiltonian is also decomposed into the sum of two terms,  
\begin{equation*}
  \hat{H}_0 
  = 
  \hat{H}'_0 + \hat{H}''_0, 
  \quad 
  \hat{H}'_0
  =
  \sum_{1\leq \alpha<\beta \leq r}
  J_{\alpha\beta}\sigma_3^{(\alpha)}\sigma_3^{(\beta)},\; 
  \hat{H}''_0
  =
  \sum_{r+1\leq \alpha<\beta \leq n}
  J_{\alpha\beta}\sigma_3^{(\alpha)}\sigma_3^{(\beta)}. 
\end{equation*}
Since $S_1$ is connected, in the same procedure as taken 
in proving Prop. \ref{Proposition vi}, 
the Lie algebra ${\cal L}_{1}$ generated from $-i\hat{H}'_0$ and 
$\{-i\hat{H}_\alpha,-i\hat{H}_{n+\alpha}\}_{\alpha=1,\cdots,r}$ 
proves to be 
${\cal L}_{1} = \mathfrak{su}(2^r)\otimes I^{\otimes(n-r)}$. 
In the same manner, the operators 
$-i\hat{H}''_0$ and 
$\{-i\hat{H}_{r+\alpha},-i\hat{H}_{n+r+\alpha}\}_{\alpha=1,\cdots,n-r}$ 
generate the Lie algebra 
${\cal L}_{2} = I^{\otimes r}\otimes \mathfrak{su}(2^{n-r})$.  
Since two operators each of which is 
in the respective Lie algebras 
${\cal L}_{1}$ and ${\cal L}_{2}$ 
commute, the Lie algebra generated by the whole operators 
$\{-i\hat{H}_\mu\}_{\mu=0,1,\cdots,2n}$ splits into the 
direct sum ${\cal L}_{1}\oplus {\cal L}_{2}$, 
which is a subalgebra 
of $\frak{su}(2^n)$ but not equal to the whole $\mathfrak{su}(2^n)$. 
We may generalize this fact to the case where the spin graph is 
broken up into more than two disjoint subgraphs, 
but the generalization is easy to perform. 
So far we have shown the following

\begin{Prop}
\label{Proposition vii}
If the spin graph associated with the drift Hamiltonian is 
disconnected, the NMR system \eqref{Schrodinger equation 0} is not controllable. 
\end{Prop}

From Props \ref{Proposition vi} and \ref{Proposition vii}, 
we obtain the following theorem;

\begin{Thm}
\label{Thm: controllability}
The NMR system \eqref{Schrodinger equation 0} with 
the drift and control Hamiltonians \eqref{Hamiltonians} is controllable, 
if and only if the spin graph associated 
with the drift Hamiltonian is connected.  
\end{Thm}


\begin{thebibliography}{99}
\bibitem{AD'A02} 
F. Albertini and D. D'Alessandro, 
The Lie algebra structure and nonlinear controllability of 
spin systems, 
Linear Algebra Appl., {\bf 350}, 213-235 (2002). 

\bibitem{AD'A03}
F. Albertini and D. D'Alessandro, 
Notions of controllability for quantum mechanical systems, 
IEEE Trans. Automat. Control, {\bf 48}, no. 8, 1399-1403 (2003).

\bibitem{Alt02}
C. Altafini, 
Controllability of quantum mechanical systems by root space 
decomposition of $\mathfrak{su}(N)$, 
J. Math. Phys., {\bf 43}, 2051-2132 
(2002). 

\bibitem{BLT03}
E. Briand, J-G. Luque, and J-Y. Thin, 
A complete set of covariants of the four qubit system, 
J. Phys. A: Math. Gen., {\bf 36}, 9915-9927 (2003). 

\bibitem{DA00}
D. D'Alessandro, 
Topological properties of reachable sets and the control of 
quantum bits, 
Systems \& Control Lett., {\bf 41}, 213-221 (2000). 

\bibitem{D'A01}
D. D'Alessandro, 
Small time controllability of systems on compact Lie groups 
and spin angular momentum, 
J. Math. Phys., {\bf 42}, 4488-4496 (2001). 

\bibitem{DAR06}
D. D'Alessandro and R. Romano, 
Decompositions of unitary evolutions and entanglement dynamics of 
bipartite quantum systems, 
J. Math. Phys., {\bf 47}, 082109 (2006). 

\bibitem{DDAS08}
M. Da\v{g}li, D.  D'Alessandro and J. D. H. Smith, 
A general framework for recursive decompositions of unitary quantum evolutions, 
J. Phys. A: Math. Theor., {\bf 41}, 155302 (2008)

\bibitem{Ema04}
C. Emary, 
A bipartite class of entanglement monotones for $N$-qubit pure states, 
J. Phys. A: Math. Gen., {\bf 37}, 8293-8302 (2004).

\bibitem{GRB}
M. Grassel, M. R{\o}tteler, and T. Beth, 
Computing local invariants of quantum bit system, 
Phys. Rev. {\bf A}, 58, 1833-1839 (1998). 

\bibitem{Iw07-1}
T. Iwai, 
The geometry of concurrence as a measure of entanglement, 
J. Phys. A: Math. Theor., {\bf 40}, 1361 (2007). 

\bibitem{Iw07-2}
T. Iwai, The geometry of multi-qubit entanglement, 
J. Phys. A: Math. Theor., {\bf 40}, 12161 (2007). 

\bibitem{IHM08}
T. Iwai, N. Hayashi and K. Mizobe, 
The geometry of entanglement and Grover's algorithm, 
J. Phys. A: Math. Theor., {\bf 41}, 105202 (2008). 

\bibitem{JS72}
V. Jurdjevic and H. J. Sussmann, 
Control systems on Lie groups, 
J. Diff. Eq., {\bf 12}, 313-329 (1972). 

\bibitem{KBG01}
N. Khaneja, R. Brockett, and S.J. Glaser, 
Time optimal control in spin systems, 
Phys. Rev. A, {\bf 63}, 032308 (2001). 

\bibitem{KG01}
N. Khaneja and S. J. Glaser, 
Cartan decomposition of $\SU{2^{n}}$ and control of spin systems, 
Chem. Phys., {\bf 267}, 11-23 (2001). 

\bibitem{KGB02}
N. Khaneja, S. Glaser, and R. Brockett, 
Sub-Riemannian geometry and time optimal control of three spin systems: 
Quantum gates and coherence transfer, 
Phys. Rev. A, {\bf 65}, 032301 (2002). 

\bibitem{Le04}
P. L\'evay, 
The geometry of entanglement: metrics, connections and the geometric phase, 
J. Phys. A: Math. Gen., {\bf 37}, 1821-1841 (2004). 

\bibitem{Lev05}
P. L\'evay, 
On the geometry of a class of $N$-qubit entanglement monotones, 
J. Phys. A: Math. Gen., {\bf 38}, 9075-9085 (2005). 

\bibitem{Le05}
P. L\'evay, 
Geometry of three-qubit entanglement, 
Phys. Rev. A, {\bf 71}, 012334 (2005). 

\bibitem{Le06}
P. L\'evay, 
On the geometry of four-qubit invariants, 
J. Phys. A: Math. Gen., {\bf 39}, 9533-9545 (2006). 

\bibitem{MD01}
R. Mosseri and R. Dandoloff, 
Geometry of entangled states, Bloch spheres and Hopf fibrations, 
J. Phys. A: Math. Gen., {\bf 34}, 102431-10252 (2001). 

\bibitem{RSD95}
V. Ramakrishna, M.V. Salapaka, M. Dahleh, H. Rabviz, and A. Peirce, 
Controllability of molecular systems, 
Phys. Rev. A, {\bf 51}, 960-966 (1995).

\bibitem{Rom07}
R. Romano, 
Entanglement magnification induced by local manipulations, 
Phys. Rev. A, {\bf 76}, 042315 (2007). 

\bibitem{RD'A06}
R. Romano and D. D'Alessandro, 
Incoherent control and entanglement for two-dimensional 
coupled systems, 
Phys. Rev. A, {\bf 73}, 022323 (2006). 

\bibitem{Sud}
A. Sudbery, 
On local invariants of pure three-qubit states, 
J. Phys. A: Math. Gen., {\bf 34}, 643-652 (2001). 

\bibitem{ZVSW03}
J. Zhang, J. Vala, S. Sastry, and K.B. Whaley, 
Geometric theory of non-local two-qubit operations, 
Phys. Rev. A, {\bf 67}, 042313 (2003). 

\end{thebibliography}
\end{document}